%% file: coop_streaming_v2.tex
\documentclass[journal,letterpaper]{IEEEtran}
\usepackage{hyperref}
\usepackage{amsmath,amssymb,amsthm}
\usepackage{graphics, subfigure}
\usepackage{tikz}
\usepackage{pgfplots}
\usepackage{algorithm, algorithmic}
\usepackage[outdir=./]{epstopdf}
\usepackage{epstopdf, xspace}
\usepackage{color}
\usepackage{multirow}
\usetikzlibrary{automata,positioning}
\usepackage{wrapfig}
\usepackage{rotating}
\usepackage{flushend}
\usepackage{bbm}

\usepackage{enumitem}
\usepackage{algorithm, algorithmic}

\newtheorem{theorem}{Theorem}

\newtheorem{lemma}{Lemma}

\newcommand{\ie}{\textit{i.e.,}\xspace}
\newcommand{\eg}{\textit{e.g.,}\xspace}

\newenvironment{thisnote}{}{}
\newcommand{\floor}[1]{\lfloor #1 \rfloor}
\usepackage{pgfplots} 
\pgfplotsset{compat=newest} 
\pgfplotsset{plot coordinates/math parser=false} 
\newlength\figureheight 
\newlength\figurewidth

\begin{document}
	\input{macros}

%

\title{GroupCast: Preference-Aware Cooperative Video Streaming with Scalable Video Coding }


 \author{Anis Elgabli, Muhamad Felemban, and Vaneet Aggarwal\thanks{The authors are with Purdue University, West Lafayette, IN 47907, email:\{aelgabli, mfelemba, vaneet\}@purdue.edu. This paper was presented, in part, at the IEEE Infocom Workshop 2018 \cite{workshop}.  
}}
 
\maketitle

\input{abstract}


\input{intro}

\input{related}

\input{model}

\input{skip_form}

\input{skip_algo}

\input{noSkip_algo}

\input{simulation_main}

\input{Conclusion}



\bibliographystyle{IEEEtran} %

\bibliography{refs,bib,multipath_feng}
\newpage
\clearpage

\appendices
\input{examples}

\input{skipproofs}

\input{symbol}
\input{sens_anal}
\end{document}

%% file: macros.tex
\definecolor{brown}{cmyk}{0,0.81,1,0.60}
\definecolor{magenta}{rgb}{0.4,0.7,0}
\definecolor{gray}{rgb}{0.5,0.5,0.5}
\definecolor{red}{rgb}{1,0,0}
\definecolor{green}{rgb}{0.5,0,0.5}
\definecolor{blue}{rgb}{0,0,1}


\ifthenelse{\isundefined{\final}} {
\newcommand{\vaneet}[1]{{\color{green}(VA: #1)}}
\newcommand{\shuai}[1]{{\color{red}(SH: #1)}}
\newcommand{\feng}[1]{{\color{blue}(FQ: #1)}}
\newcommand{\anis}[1]{{\color{brown}(AE: #1)}}
\newcommand{\shubho}[1]{{\color{magenta}(SS: #1)}}
}{
\newcommand{\vaneet}[1]{}
\newcommand{\shuai}[1]{}
\newcommand{\feng}[1]{}
\newcommand{\anis}[1]{}
\newcommand{\shubho}[1]{}
}

\newcommand{\eat}[1]{}

\newcommand{\BULLET}{\vspace{+.03in} \noindent $\bullet$ \hspace{+.00in}}

%% file: abstract.tex
\begin{abstract}

In this paper, we propose a preference-aware cooperative video streaming system for videos encoded using Scalable Video Coding (SVC) where all the collaborating users are interested in watching a video together on a shared screen. However, each user's willingness to cooperate is subject to her own constraints such as user data plans and/or energy consumption. Using SVC, each layer of every chunk can be fetched through any of the cooperating users. We formulate the problem of finding the optimal quality decisions and fetching policy of the SVC layers of video chunks subject to the available bandwidth, chunk deadlines, and cooperation willingness of the different users as an optimization problem. The objective is to optimize a QoE metric that maintains a trade-off between maximizing the playback rate of every chunk while ensuring fairness among all chunks for the minimum skip/stall duration without violating any of the imposed constraints. We propose an offline algorithm to solve the non-convex optimization problem when the bandwidth prediction is non-causally known. This algorithm has a run-time complexity that is polynomial in the video length and the number of cooperating users. Furthermore, we propose an online version of the algorithm for more practical scenarios where erroneous bandwidth prediction for a short window is used. Real implementation with android devices using SVC encoded video on public bandwidth traces' dataset reveals the robustness and performance of the proposed algorithm and shows that the algorithm significantly outperforms round robin based mechanisms in terms of avoiding skips/stalls and fetching video chunks at their highest quality possible.


\end{abstract}

%% file: intro.tex
\section{Introduction}
\label{sec:intro}

Video streaming is the dominant contributor to the cellular traffic. Currently, video content accounts for $50\%$ of cellular traffic and it is expected to account for around $75\%$ of the mobile data traffic by the year of 2020 \cite{ericsson_report}. This increase has forced service providers to enhance their infrastructures to support high-quality video streaming. Despite these efforts, users frequently experience low Quality-of-Experience (QoE) metrics such as choppy videos and playback stalls~\cite{zou2015can}. One prominent approach to improve the QoE is the use of cooperative video streaming technology~\cite{iosifidis2014enabling}, which allows to aggregate the available bandwidth for different users in order to increase the download rate; and hence increase the video playback rate. Several challenges have to be addressed when using cooperative video streaming, including finding a fetching policy among the users while using their resources efficiently. 

In Scalable Video Coding (SVC), each chunk is encoded into ordered \emph{layers}: a \emph{base layer} (BL) with the lowest playable quality, and multiple \emph{enhancement layers} ($E_1,..,E_N$) that further improve the quality~\cite{SVC_encoding}. The video player must download all layers from 0 to $i$ in order to decode a chunk up to the $i^{th}$ enhancement layer. Consequently, adaptive SVC streaming can allow playback at a lower quality if all the enhancement layers of a chunk have not been fetched before its deadline. For cooperative video streaming, when different users are willing to cooperate in order to improve the quality of the next chunk to fetch, different layers of the same chunk can be fetched through different contributors.

Cooperation among more users leads to a diversity in the user channels, and in some cases diversity in the providing carriers ({\it e.g.}, AT\&T and Verizon). Even when all users belong to the same carrier, the users of Groupcast share the resources with the other users on the same base station. If there are $N$ Groupcast users and $M$ other users of a tower, assuming equivalent channel strengths to the users, Groupcast users use $\frac{N}{M+N}$ of the base station resources which increases with $N$. We assume that the wireless link is the bottleneck, rather than the backbone network. 

To motivate our problem, consider the scenario where a group of users in a location that does not provide Internet connection, \eg a camp site, is interested in watching a football game on an High-Definition screen. The users have mobile phones with varying data plans limits, \eg unlimited, 6GB, or 3GB, and varying energy levels at the time of the game. The users are willing to participate in streaming the game using their data plans on their phones. With these assumptions, there will be preferences based on which priority sets of the users are defined. Users belong to a higher priority set can help in fetching up to a higher layer. We impose a maximum contribution limit for every user based on the remaining data on their plans. The maximum contribution indicates the maximum amount of data the user can download for the entire video. For example, certain percentage of the remaining data of the $i^{th}$ user plan can be used as a maximum contribution.

Figure~\ref{fig : ex1} illustrates an example of a setup with four users. The video is encoded into four layers, one base layer and three enhancement layers. Users 0 and 1 have unlimited data plans and have no maximum contribution limit, and thus can fetch all the layers. The two users are assigned to the top priority set. User 2 is assigned to the second priority set and has a limited data plan ($6$ GB/month). User 2 can contribute to increase the quality of the chunks up to the $1^{st}$ enhancement layer quality if the first two users fail due to their limited bandwidths. although User 2 helps in fetching chunks, she also has a maximum contribution limit of 1 GB. User 3 is assigned to the third priority set and can help in obtaining base layer if the other users are not able to obtain the chunks at that set. Further, the maximum contribution of user 3 is 500 MB. 

\begin{figure}[t!]
	\centering
	\includegraphics[ clip,  width=.48\textwidth]{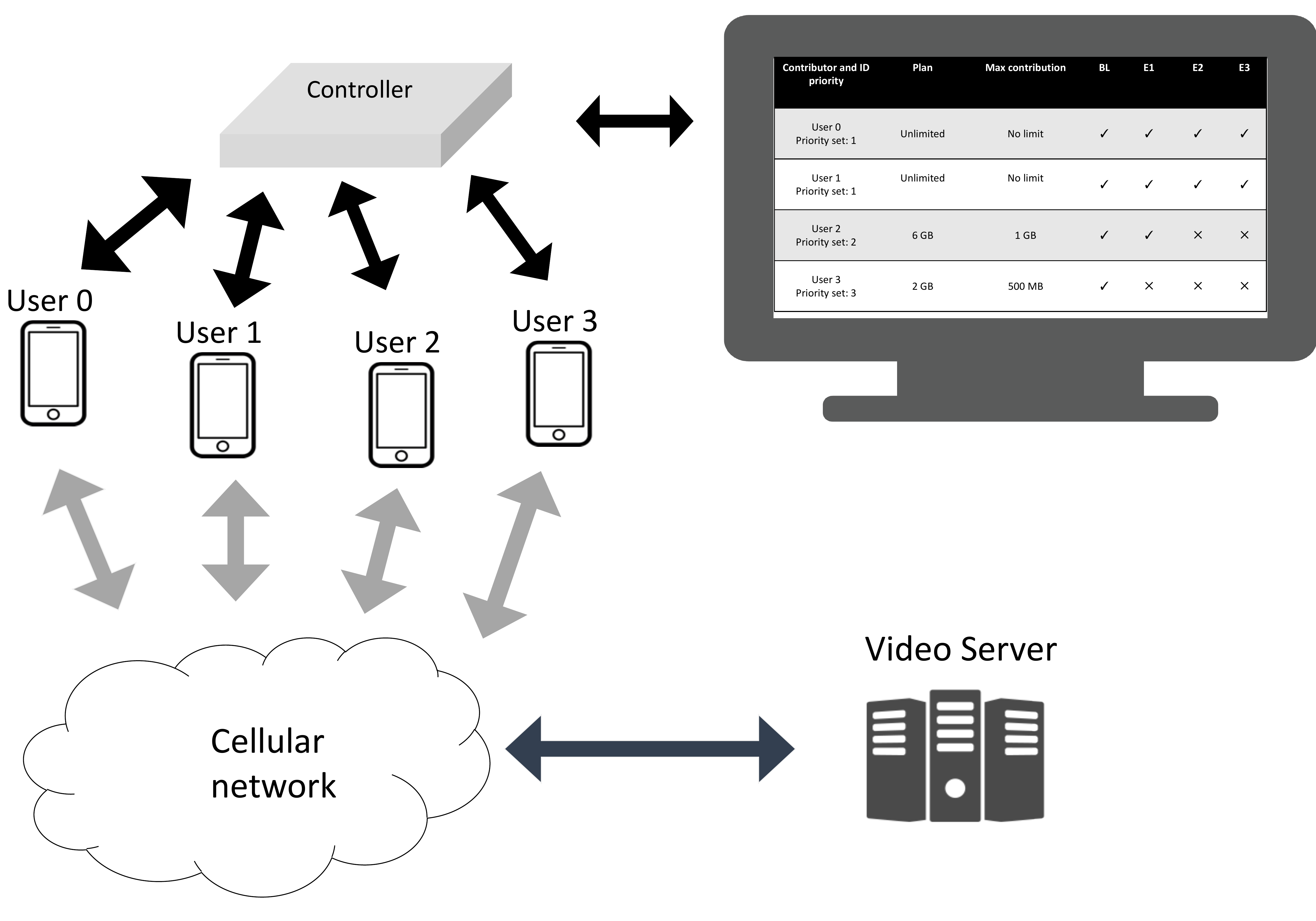}	
	
	\caption{Example to illustrate the different inputs by the different users. These inputs include the priority set, amount of data they can contribute, and the different layers they can help fetch to maximize the quality of experience for video viewing.}
	\label{fig : ex1}
\end{figure}

In this paper, we propose a novel preference-aware cooperative video streaming system that allows multiple collaborative users to watch SVC-encoded videos with high QoE on a shared screen. In particular, we propose streaming algorithms for two classes of streaming scenarios: {\em skip based} and {\em no-skip based} streaming. In both scenarios, each chunk is associated with a deadline. However, in the skip based streaming (real-time streaming), the chunks not received by their their deadlines are skipped. For no-skip based streaming, if a chunk cannot be received by its deadline, it will not be skipped; instead, a stall (re-buffering) will incur until it is fully received. Although the system and the problem formulation are described for SVC, they are applicable to any layered coding technique. Therefore, videos encoded using Scalable High efficiency Video Coding (SHVC), a scalable extension of H.265/HEVC \cite{tech2016overview}, can  be cooperatively streamed using the proposed approach.

An optimization problem is formulated when the bandwidth can be obtained offline for both the scenarios.  The optimization problem optimizes a QoE metric that maximizes a concave utility function of the chunk qualities and minimizes skip/stall duration while respecting the priority of the users (users belong to lower priority sets should be used less) and their imposed maximum contributions. Since the proposed optimization problem is non-convex, optimal low-complexity algorithms are not straightforward. Therefore, we provide efficient algorithms which have a polynomial run-time complexity with respect to the number of contributing users and fetched layers. Moreover, we provide some guarantees of the proposed algorithm. 
 
The proposed algorithms are further extended to the practical online case where the bandwidth is only known for a short time ahead and can be erroneous. Any bandwidth prediction method like the crowd-sourcing, or harmonic mean prediction method \cite{MPC,chen2016msplayer}, can be used which provides the bandwidth prediction with some errors.

The proposed algorithm is implemented and tested using android devices with real SVC encoded videos and real bandwidth traces. The algorithm is shown to significantly outperform the considered baselines. The  contributions of this paper are summarized as follows. 

\BULLET  We introduce a video streaming system, called GroupCast, that integrates the preference of the user in the cooperative streaming of SVC-encoded videos.


\BULLET We propose low-complexity, window-based algorithms for both the  {\em skip based} and {\em no-skip based} streaming scenarios.

\BULLET We provide theoretical guarantees of the proposed algorithm. 

\BULLET  We evaluate the proposed algorithms with a real implementation using  android devices, SVC-encoded videos, and public bandwidth traces. The algorithm is shown to significantly outperform the baseline algorithms.

The rest of the paper is organized as follows. Section \ref{sec:sysModel} describes the system model. Section \ref{sec:problem} describes the problem formulation for skip based cooperative video streaming. Further, a set of polynomial-time algorithms are provided in section~\ref{skipalgoEqual} for this non-convex problem. No-skip scenario is considered in section~\ref{no_skip}. Section \ref{sec:eval} presents the trace-driven evaluation results with comparison to different baselines. Section \ref{sec:concl} concludes the paper. 

%% file: related.tex
\section{Related Work}
\label{sec:related}

{\bf Adaptive Bit Rate (ABR) Streaming for single user. } Commercial systems such as  Apple's HLS~\cite{HLS}, Microsoft's Smooth Streaming~\cite{SS}, and Adobe's HDS~\cite{HDS} are different variants of ABR streaming. Various approaches for making ABR streaming decisions have been investigated, including control theory~\cite{MPC,Miller15}, Markov Decision Process~\cite{Jarnikov11}, machine learning~\cite{Claeys14}, client buffer information~\cite{BBA}, and data-driven techniques~\cite{Liu12,C3,CS2P}.

In a single user video streaming scenario, the rate adaptation techniques can be classified into: (i) buffer-based rate adaptation techniques, in which the quality decisions of the next few chunks to fetch is decided based on the current buffer occupancy (\eg ~\cite{BBA}), (ii) prediction-based techniques in which a prediction method such as crowd-sourcing~\cite{riiser2013commute,GTube} or harmonic mean~\cite{MPC} is used to predict the bandwidth of the next few seconds and decide the quality of the next few chunks to fetch, (iii) a combination of buffer and prediction in which the current buffer occupancy and the predicted bandwidth are incorporated in order to decide the quality of the next few chunks~\cite{MPC}. The third class is shown to outperform the first two classes~\cite{MPC}. In our proposed formulation, we incorporate both buffer (chunk deadlines) and prediction into our quality decisions.

{\bf Adaptive SVC streaming for single user. } SVC received the final approval to be standardized as an amendment of the H.264/MPEG-4 AVC (Advanced Video Coding) standard
in 2007~\cite{svcrelease}.  The work~\cite{SVCDataset} published the first dataset and toolchain for SVC.
Some prior work~\cite{Andelin12,Sieber13} proposed new rate adaptation algorithms for SVC that alternate between prefetching future base layers and backfilling current enhancement layers. Finally, \cite{AnisTONSingle} proposed an optimization based rate adaptation technique for single link SVC-encoded videos in which a QoE metric that maintains tradeoff between maximizing the quality of each chunk and ensuring fairness among all chunks for the minimum stall/skip duration is optimized.

{\bf Multi-Path video streaming.} For multi-path streaming in which a single client is opening parallel connections with a server using orthogonal links,\ie, LTE and WiFi, Multi-path TCP (MPTCP~\cite{mptcp}) is proposed. However, MPTCP does not work across users and have some other implementation and security issues~\cite{chen2016msplayer,deconinck16:pam}. The authors of~\cite{chen2016msplayer} proposed a heuristic approach for Multi-Path rate adaptation in which multiple servers, single client with two links (LTE and WiFi), and AVC encoded videos are considered. However, this rate adaptation technique considers only two paths, and is proposed for ABR based streaming. Thus, it does not exploit the flexibility of SVC where a fine-grained decision per layer can be taken. Furthermore, these work do not account for maximum contribution limit on one/both of the paths. Finally, \cite{DBLP:journals/corr/abs-1801-01980} considered SVC multipath preference aware streaming for SVC but for only 2 links without imposing maximum contribution per link.   

{\bf Cooperative video streaming.} For multiple users interested in watching the same video, streaming algorithms using device-to-device (D2D) connections have been studied. In \cite{stiemerling2009system}, authors proposed a Peer-to-Peer (P2P)-TV system that enables multiple contributors to download a single rate encoded video in which every contributor is assigned a chunk to fetch such that all chunks are retrieved before their deadlines. In \cite{keller2012microcast}, authors proposed network coding based cooperative video streaming for a single rate encoded video where multiple users can aggregate their bandwidth in a more reliable way and fetch the video. However, none of these efforts considers SVC-encoded videos in which every chunk is available at base layer and its quality can be further improved by downloading enhancement layers. 

The authors of \cite{grafl2013scalable} explain the role of Content-Aware Networking (CAN) to optimize network resource utilization while maintaining high QoE and QoS. The functionality of CAN is provided by enhanced network nodes, termed \textit{media-aware network elements} (MANEs), that feature feature virtualization support, content awareness, and media processing as well as buffering and caching. The proposed MANEs take advantage of SVC technology in order to provide scalable streaming for users. Although the paper has laid out the main challenges of CAN in several use cases including P2P, the paper does not provide streaming algorithms for SVC-encoded videos. In \cite{eberhard2010knapsack}, authors have proposed piece-picking algorithms to download selective pieces of distributed layered content over P2P networks before their deadlines. The proposed algorithms for skip-based streaming are provided based on Knapsack problem. However, the setup in \cite{eberhard2010knapsack} obtains content at a receiver from multiple peers while in our problem, multiple users obtain content simultaneously to display on a single screen. 

In \cite{yaacoub2013svc}, authors have considered cooperative SVC video streaming for SVC-encoded videos but with only two layers for LTE/802.11p Vehicle-to-Infrastructure Communications. In \cite{lee2011co}, authors proposed a video assignment strategy based on the deficit round robin (DRR) for combined SVC and Multiple Description Coding (MDC), termed Co-SVC-MDC-Based Cooperative video Streaming Over Vehicular Networks. In this strategy, the decision of assigning layers to users follow simple round robin strategy and does not take into account users's preferences. The authors of \cite{an2017svc} proposed an SVC cooperative video streaming for vehicular networks when there is only one relay node per receiver, \ie at most two contributors. In contrast to these works, we do not limit the number of cooperating users, and we propose an algorithm for SVC-encoded videos with any number of layers. In addition, we assume preference aware streaming from users with a maximum contribution limit per user. We consider both skip and no-skip based streaming. We formulate the cooperative video streaming problem as a non-convex optimization problem and  develop an efficient and practical algorithm for solving the proposed problem. The proposed algorithm makes a fine-grained decision per layer in order to decide if it can be fetched and to decide which user can fetch that layer. 

%% file: model.tex
\section{System Model}
\label{sec:sysModel}

{In this section, we describe the system model. A summary table for all the symbols used throughout the paper is listed in Appendix \ref{app:symbols}}. We assume that there are $U$ users and that the SVC-encoded video is divided into $C$ chunks, where every chunk is of length $L$ seconds. The video is encoded in Base Layer (BL) with rate $r_0$ and $N$ enhancement layers ($E_1, \cdots, E_{N}$) with rates $r_1, \cdots, r_{N}$,  respectively. Let $Y_n=L*r_n$ be the size of the $n^{th}$ layer, where $n \in \{0, \dots, N\}$.

The users are divided into $K$ disjoint sets, where $K \le N+1$. The number of users in the $k^{th}$ set, for $k\in \{1, \cdots K\}$, is denoted by $U_k$. Notice that $U = \sum_{k=1}^K U_k $. We denote the maximum layer that the users in the $k^{th}$ set can fetch as $N_k$, where $N_k\in \{0,\cdots, N\}$. We assume that the maximum layer that can be fetched by the $(k+1)^{th}$ set is less than the maximum layer that can be fetched by the $k^{th}$ set, \ie $N_1 > N_2 > \dots > N_K$. To illustrate, assume that $K=2$, \ie the users are divided into two sets. Therefore, if the users in the first set, \ie $k=1$, can fetch, for example, up to the $3^{rd}$ layer, \ie $N_1 = 3$, users in the second set, \ie $k=2$, can at most fetch up to the $2^{nd}$ layer, \ie $N_2 = 2$. 

We assume that each user, $u$, in the $k^{th}$ set has a maximum download contribution of $\eta_{u}^k$.
Furthermore, we assume a strict priority among users in different sets. In essence, a user in the $k^{th}$ set fetches chunks at layers from $\{0, \cdots, N_k\}$ only if the users in the higher priority sets, \ie sets $1$ to $k-1$, are unable to fetch layers from $\{0, \cdots, N_k\}$ without violating the deadlines. In the rest of this paper, we use the terms \textit{user} and \textit{link} interchangeably.

The packing of the users into sets can be done automatically based on the available data plans. The available data plans can be quantized to get the set assignment. Then, a particular percentage of the available data plan can be chosen as the maximum contribution limit. Even though the set assignments and the contribution limits can be set automatically, the users have the flexibility to modify these limits. The designer may, however, limit the possibility of decreasing the priority level or the maximum contribution limit to avoid greedy behavior of certain users. However, the decrease of the maximum contribution based on available battery charge at the user may be allowable. In this paper, we assume that the user priority levels and the maximum contribution limits are known.

We consider two streaming scenarios: the skip based streaming and the no-skip based streaming. For the skip based streaming, the video is played with an initial start-up, \ie buffering, delay of $s$ seconds. There is a playback deadline for each of the chunks, where chunk $i$ needs to be downloaded by $deadline(i) = s+(i-1) L$ seconds. The chunks that cannot be received before their respective deadlines are skipped. For the no-skip based streaming, it also has a start-up delay of $s$ seconds. However, if a chunk cannot be downloaded before its deadline, it will not be skipped. Instead, a stall, \ie re-buffering, will occur. As a result, the video will pause until the chunk is fully downloaded. Let the total stall duration from the start till the play-time of chunk $i$ be $d(i)$. Therefore, the deadline of any chunk $i$ is $deadline(i)=(i-1)L+s+d(i)$. In both scenarios, the objective of the scheduling algorithm is to determine which layers need to be fetched for each chunk such that the skip/stall duration is minimized as the first priority and the overall playback bitrate is maximized considering the priorities and the maximum contribution of different users.

We assume that all time units are discrete and that the discretization time unit is 1 second. Let $Z_{n,i}$ denotes how much of the $n^{th}$ layer of chunk $i$ that can be fetched, \ie if the $n^{th}$ layer of chunk $i$ can be fetched, then $Z_{n,i}=Y_n$; otherwise $Z_{n,i}=0$. Let $z_{n,u}^{(k)}(i,j)$ be the amount of the $n^{th}$ layer of chunk $i$ that is fetched at time slot $j$ from the link ${\bf L}_{k,u}$, where ${\bf L}_{k,u}$ represents the $u^{th}$ user in $k^{th}$ set. The sum of $z_{n,u}^{(k)}(i,j)$ over all $u,k,j$ is $Z_{n,i}$. Let $D_{n,u}^{(k)}(i) = \sum_{j=1}^{deadline(C)} z_{n,u}^{(k)}(i,j)$ be the total amount of the $n^{th}$ layer of chunk $i$ that is fetched from the link ${\bf L}_{k,u}$. Further, let $G_n^{(k)}$ be the total amount of the $n^{th}$ layer fetched from all users in the $k^{th}$ set. Thus, $G_n^{(k)} = \sum_{u=1}^{U_k}\sum_{i=1}^{C}D_{n,u}^{(k)}(i)$. We assume that at most one link can be used to get the entire $n^{th}$ layer of chunk $i$. Therefore, for every layer $n$ of chunk $i$, there will be at most one link such that $D_{n,u}^{(k)}(i) > 0$. Let $B_u^{(k)}(j)$ be the available bandwidth of the $u^{th}$ link in the $k^{th}$ set.


%% file: skip_form.tex
\section{Skip Based Streaming: Problem Formulation}
\label{sec:problem}

In this section, we present the problem formulation for the skip based scenario. The key objectives of the problem are to \emph{(i)} minimize of  the number of skipped chunks, \emph{(ii)} maximize of the  playback rate of the video,  and \emph{(iii)} minimize of the quality changes between the neighboring chunks. As a result, the obtained optimal fetching policy is the one that minimizes the total number of skips as the highest priority, while preferring chunks at the $n^{th}$ layer quality level over increasing the quality of some chunks to higher layers at the cost of dropping the quality of other chunks. Such strategy avoids unnecessary layer switchings and leads to a smoother and more eye-pleasant playback of the video. We first describe the proposed problem formulation of the offline version that assumes perfect prediction for the whole period of the video and infinite buffer capacity. These assumptions are relaxed when we describe the online algorithm in Section~\ref{sec:bw_err}.

In order to account for the three objectives with their priority orders, we maximize a weighted sum of layer sizes fetched by different users, where lower layers and users belong to higher priority sets are given higher weights. In order to do so, we introduce two-dimensional weights $\lambda_n^{k}$ that need to satisfy the following condition:  for any users sets $k$, $k^\prime \in \{(k+1)....K\}$, and $k^{\prime\prime} \in \{1....k\}$
 \begin{equation}
 \lambda_a^{(k)} Y_{a} > C \cdot \Big( \sum_{n=a}^{N}\lambda_n^{k^{\prime}}Y_n+ \sum_{n=a+1}^{N}\lambda_n^{k^{\prime\prime}}Y_n\Big) .
 \label{basic_gamma_1}
\end{equation}
The choice of $\lambda$'s that satisfies Equation \ref{basic_gamma_1} implies two requirements. First, it implies that, for any layer $a$, the layers higher than $a$ have lower utility than a chunk at layer $a$. In the case when $a=0$, the choice of $\lambda$ implies that all the enhancement layers achieve less utility than one chunk at the base layer. The use of $\lambda$ then helps in giving higher priority to fetch more chunks at the $n^{th}$ layer quality over fetching some at higher quality at the cost of dropping the quality of other chunks to below the $n^{th}$ layer quality. Second, the choice of $\lambda$ implies that the highest utility that can be achieved for every layer is when it is fetched by a user belong to the set with the highest priority level. This will obey the priority order, and it will not use a lower priority user to fetch a layer that can be fetched by a higher priority user.

Accordingly, we formulate the objective function as the following.

\begin{equation}
\sum_{k=1}^K\sum_{n=0}^{N}\lambda_n^{(k)} G_{n}^{(k)}.
\label{equ:mainObj}
\end{equation}

This objective is constrained by a set of constraints such as bandwidth, deadline, users maximum contribution, users-set assignment, and video decoding constraints. Accordingly, the proposed optimization problem can be formulated as follows.

\begin{eqnarray}
{\textbf Maximize: }\sum_{k=1}^K\sum_{n=0}^{N}\lambda_n^{(k)} G_n^{(k)} 
\label{equ:eq1}
\end{eqnarray}
subject to
\begin{eqnarray}
G_n^{(k)} &=&\sum_{u=1}^{U_k}\sum_{i=1}^{C}D_{n,u}^{(k)}(i) \forall n,k	\label{equ:c0eq1}\\
	D_{n,u}^{(k)}(i) &=& \sum_{j=1}^{(i-1)L+s} z_{n,u}^{(k)}(i,j)  \forall i,  n, k, u
	\label{equ:c1eq1}
\end{eqnarray}

\begin{eqnarray}
 \sum_{k=1}^K \sum_{u=1}^{U_k}D_{n,u}^{(k)}(i) = Z_{n,i}\quad  \forall i,  n 
\label{equ:c2eq1}
\end{eqnarray}
\begin{eqnarray}
Z_{n,i}\le \frac{Y_n}{Y_{n-1}}Z_{n-1,i}\quad  \forall i,  n>0 
\label{equ:c3eq1}
\end{eqnarray}
\begin{eqnarray}
\sum_{n=0}^N\sum_{i=1}^{C} z_{n,u}^{(k)}(i,j)  \leq B_u^{(k)}(j) \  \   \forall k,u,\text{ and } j
\label{equ:c4eq1}
\end{eqnarray}
\begin{align}
&D_{n,u}^{(k)}(i) D_{n,u'}^{(k')}(i) = 0 \forall n, i, \text{ and } (k,u)\ne (k',u') 
\label{equ:c5eq1}
\end{align}
\begin{equation}
\sum_{i} D_{n,u}^{(k)}(i) \leq \eta_u^k    \forall u,k
\label{equ:c6eq1}
\end{equation}
\begin{equation}
D_{n,u}^{(k)}(i) =0    \forall k, n > N_k
\label{equ:c6eq1_2}
\end{equation}
\begin{equation}
z_{n,u}^{(k)}(i,j) \geq 0\   \forall u, \forall i 
\label{equ:c7eq1}
\end{equation}
\begin{equation}
z_{n,u}^{(k)}(i,j)= 0\   \forall \{i: (i-1)L+s > j\}, \forall u
\label{equ:c8eq1}
\end{equation}
\begin{equation}
Z_{n,i} \in {\mathcal Z}_n\triangleq \{0, Y_n\} \quad  \forall i, n
\label{equ:c9eq1}
\end{equation}

In the above formulation,  
Constraints \eqref{equ:c0eq1}, \eqref{equ:c1eq1}, \eqref{equ:c2eq1} and \eqref{equ:c9eq1} ensure that what is fetched for layer $n$ of chunk $i$ over all links and times to be either zero or $Y_n$. Constraint  \eqref{equ:c3eq1} imposes the decoding constraint, \ie it ensures that the $n^{th}$ layer of any chunk cannot be fetched if the lower layer is not fetched.
Constraint \eqref{equ:c4eq1} imposes the bandwidth constraint of all links at each time slot $j$.
Constraint \eqref{equ:c5eq1} enforces that a chunk's layer can be fetched only over one link. Constraint \eqref{equ:c6eq1} imposes that the maximum amount of data fetched from link ${\bf L}_{k,u}$ is $\eta_u^k$. Constraint \eqref{equ:c6eq1_2} enforces that user $u$ at the $k^{th}$ set cannot be used to fetch layers higher than $N_k$. Constraint \eqref{equ:c7eq1} imposes the non-negativity of the download of a chunk and \eqref{equ:c8eq1} imposes the deadline constraint since chunk $i\in \{1, \cdots, C\}$ cannot be fetched after its deadline ($deadline(i)=(i-1)L+s$). Recall that $s$  is the initial startup delay.

The problem defined in \eqref{equ:eq1}-\eqref{equ:c9eq1} has integer constraints and a non-convex constraint, \ie Equation \eqref{equ:c5eq1}. Integer-constrained problems are in the class of discrete optimization, which are known to be NP hard in general \cite{nemhauser1988integer}. 

%% file: skip_algo.tex
\section{GroupCast Algorithm for Skip-based streaming}\label{skipalgoEqual}
 
 
In this section, we describe the proposed algorithm for skip-based streaming. To develop the algorithm, we consider the offline algorithm, \ie the bandwidth is perfectly known for the whole period of the video. We first assume one set of users, \ie $K=1$, where there is no preference among different users. We refer to this algorithm as ``Offline No-Pref GroupCast". This algorithm is, then extended to the general case in which users can belong to different sets with different priority levels. We refer to this algorithm as ``Offline Pref GroupCast".   Finally, we propose online algorithm (Online Pref/No-Pref GroupCast) in which more practical assumptions are considered such as short bandwidth prediction with prediction error and finite buffer size.
  

\subsection{Offline No-Pref GroupCast Algorithm }\label{skipalgoNoPref}

We now describe Offline No-Pref GroupCast algorithm, and for more illustration of the algorithm we have included an  example in Appendix~\ref{ex1}. In No-Pref GroupCast, We assume that there is only one set that contains $U$ users. The No-Pref GroupCast algorithm (Algorithm 1) process the layers sequentially according to their order (Line 4). The first layer to consider is the base layer where the algorithm initially finds the cumulative bandwidth of every second $j$ and user $u$ ($R^{(u)}(j)$) (Line 5). Then, it makes the decision for the base layers of all the chunks, \ie which chunks to be skipped and which to be fetched. Moreover, the algorithm decides which user is used to fetch the base layer of every chunk that has been decided to be fetched. The algorithm performs forward scan and finds the maximum number of base layers that can be fetched before the deadline of every chunk $i$ ($V_{0,i}$). The maximum number of base layers that can be fetched before the deadline of the $i^{th}$ chunk is: $V_{0,i}=\sum_{u=1}^{(U)} \floor{\min(\frac{\eta_u}{y_0},\frac{R^u(deadline(i))}{y_0})}$(Line 6). Let $skip(i)$ be the total number of skips before the deadline of the chunk numbered $i$. Therefore, if $V_{0,i}$ is less than $i-skip(i-1)$ at the deadline of the $i^{th}$ chunk, there must be a skip/skips, and the total number of skips from the start until the deadline of chunk $i$ will be equal to $skip(i)=skip(i-1)+1$(Lines 9-11). If there are $A$ skips, the algorithm will always skip the first $A$ chunks since they are the closest to their deadlines. Thus, skipping them will result in a bandwidth that can be used by all of the remaining chunks to increase their quality to the next layer. This choice maximizes the total available bandwidth for the later chunks. Note that the number of skips could increase if the maximum contribution of all users is used before downloading all chunks.

In the second step, the algorithm performs backward scan per chunk (Lines 15-24) starting from the first chunk that was decided to be fetched by calling the Backward Algorithm (Line 20, Algorithm 2). The backward algorithm simulates fetching every chunk $i$ starting from its deadline and by every user that has $min(R^{(u)}(deadline(i)),\eta_u^k) \geq Y_0$(Lines 16, 18, and 20).

Before we describe the second step into details, we define some parameters. Let $\alpha_i^n(u)$ be the amount of bandwidth used to fetch the layer $n$ of chunk $i$ by user $u$ before the deadline of chunk numbered $i-1$. Let $\zeta_i^n(u)$ be the cost of fetching the layer $n$ of chunk $i$ by user $u$, and $\zeta_i^n(u)$ can be found as follows:

\begin{equation}\label{equ:cost}
\left\{\begin{array}{l}
\zeta_{i}^n(u)=\alpha_i^n(u), \text{if $Y_n$ can be fetched by user $u$}\\
\zeta_{i}^n(u)=\infty,  \text{ otherwise}\\
\end{array}\right.
\end{equation}

With $\zeta_i^n(u)$ being defined, we describe the second step of the algorithm. The algorithm performs backward scan per chunk starting from the first chunk  by calling Algorithm 2 (line 20). The backward scan simulates fetching every chunk $i$ starting from its deadline and by every user. The algorithm computes the the cost of fetching the base layer of chunk $i$ by every user $u$ ($\zeta_i^n(u)$). The user choice that minimizes the cost is chosen to fetch the base layer of chunk $i$. Note that the link over which the chunk $i$ will be fetched is the one that gives the maximum amount of total available bandwidth over all links before the deadline of chunk numbered $i-1$.  For example, assume that there are only 2 users, and consider that fetching the $i^{th}$ chunk by user 1 results in using  $x$ amount of the bandwidth before the deadline of $i-1^{th}$ chunk while fetching the $i^{th}$ chunk by user 2 results in using  $y$ amount of the bandwidth before the deadline of $i-1^{th}$ chunk. Then, the first user will be chosen to fetch the chunk $i$ if $x<y$. The objective is to free as much as possible of early bandwidth since it will help more chunks to fetch their higher layers because it comes before their deadlines.

{\bf Enhancement layer modifications}: The algorithm proceeds by performing forward-backward scan per enhancement layer in order. The bandwidth is now modified to be the remaining bandwidth after excluding whatever has been reserved to fetch lower layers (Line 6 in Algorithm 2).  Also note that the $n^{th}$ layer for a chunk is not considered if its $n-1^{th}$ layer is not decided to be fetched (Line 17 in Algorithm \ref{algo:mpsvc}).

\begin{figure}
    \scalebox{0.95}{
		\begin{minipage}{\linewidth}
			\begin{algorithm}[H]
			\small
				\begin{algorithmic}[1]
				\STATE {\bf Input:}  $\mathbf{Y}=\{Y_n\forall n\}$, $L$, $s$, $C$, $\eta_u$, $\mathbf{B}^{(u)}=\{B^{(u)}(j)\forall j\}$, $u=1,\cdots,U$. 
				\STATE {\bf Output:} $I^{(u)}_n$, $u=1,\cdots,U$: set containing the indices of the chunks that can have their $n$th layer fetched by user $u$.	
								
   \STATE  $deadline(i)=(i-1)L+s \quad \forall i$
  
    \FOR{each layer $n = 0, \cdots, N$}   
     \STATE $R^{(u)}(j)=\sum_{j^\prime=1}^{j} B^{(u)}(j^\prime)$,$\forall j,u$
	\STATE $V_{n,i}=\sum_{u=1}^U\floor{\min(\frac{\eta_u}{y_n},\frac{R^u(deadline(i))}{y_n})}$, $\forall i$
  		 \STATE $skip(0)=0$
     \FOR{$i=1:C$}
       \IF{$V_{n,i} < i-skip(i-1)$ or ($i \notin  I_{(n-1)}^{(1)}$ and,$\cdots$, $i \notin I_{(n-1)}^{(U)}$)}
      \STATE $skip(i)=skip(i-1)+1$
      \ENDIF
      \ENDFOR
        \STATE Skip the first $skip(C)$ chunks.
        \STATE $i^\prime=skip(C)+1$: the index of the first chunk to fetch

    \FOR{$i=i^\prime:C$}
    \STATE $j_u =deadline(i),u=\in \{1,\cdots,U\}$
    	\IF{($n=0$ or $i \in I_{n-1}^{(1)}$ or $\cdots$, $i \in I_{n-1}^{(U)})$}
							     \STATE ${B2}_{u}={B}_{u}$, $t=deadline(i-1)$, $\eta_{t,u}=\eta_u$
      				\STATE $[B2_{u},\zeta_i^n(u), \eta_{t,u}]=$ Backward$(u,j_u, \mathbf{B2}_{u}, Y_n, t, \eta_{t,u}) \forall u$
			\STATE $u_1=\arg\min (\zeta_i^n)$
      				\STATE $I_n^{(u_1)}=I_n^{(u_1)} \cup i$, $\mathbf{B}_{(u_1)}=B2_{(u_1)}$, $\eta_{u}=\eta_{t,u}$

	\ENDIF
      \ENDFOR
     \ENDFOR
   				\end{algorithmic}
				\caption{Offline No-Pref GroupCast Algorithm }\label{algo:mpsvc}
			\end{algorithm}
		\end{minipage}
	}
	\end{figure}

\begin{figure}
		\vspace{-.1in}
		\begin{minipage}{\linewidth}
			\begin{algorithm}[H]
				\small
				\begin{algorithmic}[1]
				\STATE {\bf Input:} $u,j, B2, Y_n, t, \eta_{t,u}$
					\STATE {\bf Output:} $\zeta_i^n(u)$the cost of fetching layer $n$ of chunk $i$ by user $u$, $B2$ is the residual bandwidth after fetching chunk $i$, $\eta_{t,u}$: the remaining contribution of user $u$
					
    \STATE {\bf Initialization:}
   $\zeta_i^n(u)=0$
     \WHILE {($Y_n > 0$)}
            \STATE $fetched=min(B2(j), \eta_{t,u}, Y_n)$
             \STATE $B2(j)=B2(j)-fetched$, $Y_n=Y_n-fetched$, 
             \STATE $\eta_{t,u}=\eta_{t,u}-fetched$
            \STATE {\bf if } {$(j \leq t)$} {\bf then }  $\zeta_i^n(k)=\zeta_i^n(u)+fetched$
	   \STATE {\bf if } {$(B2(j)=0)$} {\bf then }  $j=j-1$            	
	    \STATE {\bf if } {$j < 1$ and $Y_n > 0$} {\bf then }  $\zeta_i^n(u)=\infty$, break
\ENDWHILE
				\end{algorithmic}
				\caption { Backward Algorithm }
			\end{algorithm}
		\end{minipage}
		\vspace{-.2in}
	\end{figure}

{\bf Chunks' Download: }During the actual fetching of the chunks,  each link fetches the layers in order of the chunks. In other words, the $n$-th layer of chunk $i$ is fetched before the $m$-th layer of chunk $j$ if $i < j$. Moreover, for the same chunk, the layers are fetched according to their order. For example, the base layer is the first layer to download since if it is not received by the chunk's deadline, the chunk can not be played. Moreover, none of the higher layers received can be decoded if the base layer is not received by the chunk's deadline.

{\bf Complexity Analysis}:  The algorithm sequentially decides each layer. For each layer $n$, the algorithm first finds $r^{(u)}(j)$, the cumulative bandwidth of every link $u$ and time slot $j$,  which is linear with respect to every link. Thus, it has a run-time complexity $O(U\cdot deadline(C))$. Then, it  performs forward scan that has a run-time complexity $O(U\cdot deadline(C))$. Finally, a backward scan on each link at each time. Since the complexity of backward algorithm is linear in $C$ with respect to every link, the run-time complexity for a layer is $O(U\cdot C^2)$. Thus, the overall run-time complexity is $O(U\cdot N\cdot C^2)$. 
 
 We can show that this algorithm is optimal in two cases. First, if all chunks are encoded in base layer only, and secondly when the video is encoded into multiple layers but there is only one user (single user system). The second follows by an adaptation of results in \cite{AnisTONSingle}. The detailed proofs are omitted here since they are not  interesting special cases of the problem. 
 

Even though we do not show the optimality of No-Pref GroupCast in general, we note that this algorithm  minimizes the number of the $n$-th layer skips given the decisions of all lower layers. This result is formally given in the following theorem. 

\begin{theorem}
Given size decisions up to $n-1^{th}$ layer ($Z_{0,i}, \cdots,Z_{n-1,i}$, for all $i$), remaining bandwidth, and $deadline(i)$ for every chunk $i$, the proposed No-Pref GroupCast algorithm achieves the minimum number of $n^{th}$ layer skips (or obtains the maximum number of chunks at layer $n$) as compared to any feasible algorithm which fetches the same layers of every chunk up to layer $n-1$.
 \label{them:skip1} 
\end{theorem}
\begin{proof}
Proof is provided in Appendix \ref{apdx_skip1}.
\end{proof}
\subsection{Offline Pref-GroupCast Algorithm}
\label{sec:pref}

\begin{figure}
	\vspace{-.1in}
 \scalebox{0.95}{	\begin{minipage}{\linewidth}
		\begin{algorithm}[H]
			\small
			\begin{algorithmic}[1]
					\STATE {\bf Input:}  $\mathbf{Y}=\{Y_n\forall n\}$, $L$, $s$, $C$, $\mathbf{B}_{u}^{(k)}=\{B_{u}{(k)}(j)\forall j\}$, $k=1....K$, $u=1,..,U_k, \forall k$. 
				\STATE {\bf Output:} $I^{(u,k)}_{n}$, $n=0,...,N$: set containing the indices of the chunks that can have their $n$th layer fetched by user $u$ in set $k$.	
	\STATE $\chi=\{1,....,K\}$: set of all user sets
			\STATE $lowest=K$
			\STATE $p=0$
			\WHILE{  $2 \in \chi$}			
				\STATE m=max layer index that users belong to set $lowest$ can fetch
				\STATE $[I^{(u,k)}_{n},\mathbf{B}_{u}^{(k)}, n=p\cdots m]=$No-Pref GroupCast($\mathbf{Y}$, $L$, $s$, $\mathbf{B}_{u}^{(k)}$, $k=1 \cdots lowest$, $u=1 \cdots U_k$)

				\STATE $\mathbf{Y2}=\{Y_{n,i}\forall n,i \in I_n^{(lowest)}\}$		
				
				\STATE $[I_n^{(k^{\prime})},\mathbf{B}_{u}{(k^\prime)}, n=0\cdots n_2, k^\prime=1 \cdots lowest-1]=$No-Pref GroupCast($\mathbf{Y2}$, $L$, $s$, $\mathbf{B}_{u}^{(k^{\prime})}$)
				
				\FOR{$i=1:C$}
				\FOR {$u=1:U_lowest$}
					\IF{$(i \in I_n^{(u,lowest)}$ and $i \in I_n^{(u \in k^{\prime},k^{\prime})})$}
						\STATE $I_n^{(u, lowest)}=I_n^{(u,lowest)}-\{i\}$
						\ENDIF
				\ENDFOR
				\ENDFOR
		\STATE $\chi=\chi-\{lowest\}$
				\STATE $lowest=lowest-1$
				\STATE $p=m+1$
			\ENDWHILE

			\STATE $[I_n^{(u,1)},\mathbf{B}_{u}^{(1)}, n=p\cdots N]=$No-Pref GroupCast($\mathbf{Y_n}, n\geq p$, $L$, $s$, $\mathbf{B}_{u}^{(1)}$)
			\end{algorithmic}
			\caption{Offline Pref GroupCast Algorithm }
			\label{pref-algo}
		\end{algorithm}
	\end{minipage}
}
\end{figure}

We now describe Pref GroupCast algorithm, and for more illustration of the algorithm we have included an  example in Appendix~\ref{ex2}. In this section, we consider $K$ sets in which the users belong to the first set, \ie users $1$ to $u^\prime$, have unlimited data plan and can contribute as much as their bandwidths allow. On the other hand, the remaining users, \ie users $u^\prime+1$ to $U$ have different preferences and can belong to different sets. For example, some of them can contribute in fetching up to the $n^{th}$ enhancement layer, and some others can only help in avoiding skips due to their limited plans. The later ones will belong to the set numbered $K$ since sets are ordered according their priority levels (willingness to contribute). The proposed Offline Pref GroupCast algorithm is listed in Algorithm \ref{pref-algo}. 

The first step in the algorithm is calling the No-Pref GroupCast considering the lowest set of layers that all users can fetch (Line 7-8). We denote these layers by layers $0$ to $m$. The objective of the first call of Offline No Pref-GroupCast is to find the maximum layer less than or equal to $m$ that can be fetched for every chunk. The initial fetching policy obtained from this run is not final since the preference has not been taken into consideration. Recall that the lower priority links will not be used to fetch even the the lowest set of layers, \ie $0,\cdots,m$, if all higher priority links can do so. In the second step, No-Pref GroupCast is re-run again for layers $0$ to $m$ that were initially decided to be fetched by the users belong to the lowest priority set (set $K$) but without considering this set (Lines 9-10). The objective of the second call of No-Pref GroupCast is to move as much as possible of these layers to users belong to the higher priority sets. In the next step, the lowest priority set among the remaining sets is excluded (Line 18), and the same process is repeated. At the last step of the algorithm, No-Pref GroupCast Algorithm is run for the remaining layers considering only the users belonging to the highest priority set (set 1). 

\subsection{Online Adaptations for GroupCast}
\label{sec:bw_err}

For the algorithms described in Sections \ref{skipalgoNoPref} and \ref{sec:pref}, we assumed a perfect bandwidth prediction, and the client buffer capacity is unlimited. However, practically, the prediction will not be perfect, and the client buffer might be limited. In fact, even if the buffer is unlimited the video content itself may not be available for a few chunks ahead. Therefore, the buffer constraint is considered to account for both scenarios, \ie limited client buffer and the availability of future chunks. In this section, we will use an online algorithm that will obtain prediction per user for a window of size $W$ chunks ahead and make decisions based on the prediction. The $W$ chunks start from chunk $i$ that has the current time $\delta$ higher than its deadline. \ie $deadline(i)=$current time + $\delta$. Typically, $\delta$ is 1-2 seconds. We have some marging $\delta$ to avoid making a decision for a chunk at its deadline and end up not having its decided layers fully downloaded by its deadline due to prediction error. 

There are multiple ways to obtain the prediction. Our approach is to use the harmonic mean of the download time of the past $\beta$ layers to predict the future bandwidth~\cite{chen2016msplayer}. The decisions are re-computed for the chunks that have not yet reached their deadlines periodically every $\alpha$ seconds. Typically, we set $\alpha$ to be 2-3 seconds. To account for the buffer, we assume that $WL+s$ is no more than the buffer duration. 

For the next $W$ chunks, the algorithm described in Sections \ref{skipalgoNoPref} and \ref{sec:pref} are run to find the quality using the predicted bandwidth. These  $W$ chunks can then  be fetched according to the decision generated by the algorithm. After  $\alpha$ seconds, the optimization problem is re-run, and the decisions are re-considered and updated. The fetching policy is updated on fly without interrupting the download of the current set of layers.

Since the maximum contributions are defined for the whole period of the video, the controller is required to efficiently use the maximum contribution per user over the entire duration of the video. In essence, at each time $c\alpha$, where $c \in \{1,2, \dots \}$, the controller decides the maximum contribution of user $u$ for the current $W$ chunks. Let $\eta_{u}$ be the maximum contribution of user $u$ for the entire video duration of $T$, \ie $T = (i-1)L+s+d(C)$. Further, let $f_{u,c}$ be the amount of content fetched by user $u$ before the time $c\alpha$. Then, the maximum contribution of user $u$ for the window starting at $c\alpha$, denoted  as $\tilde{\eta}_{u,c}$,  is given as

\begin{equation}
\tilde{\eta}_{u,c}=\frac{\min(WL+c\alpha L,T)}{T}\cdot \eta_{u} - f_{u,c}.
\label{equ:conPerW2}
\end{equation}

Accordingly, the maximum contribution per user is fairly distributed over time. In other words, users who contributed less than the assigned maximum contribution in the previous windows will contribute more in the current window. The initial contribution is set $\frac{WL}{T}\cdot \eta_u$ in the first window for all users.

%% file: noSkip_algo.tex
\section{No Skip Based  Streaming Algorithm}\label{no_skip}

In no-skip streaming, the video player stalls the video and continues downloading the chunk instead of skipping it when the deadline of the chunk is missed. The objective is to maximize the average quality while minimizing the stall duration. The resulted objective function is slightly different from Equation~\eqref{equ:eq1} because skipping the base layer is not allowed. However, higher layers can be skipped.
All constraints are the same as skip based optimization problem except that we add a new  constraint, \ie Constraint~(\ref{equ:c10eq2}), to enforce that $Z_{0,i}$ to be equal to the size of BL  of the $i^{th}$chunk $\forall i$. Note that $Z_{0,i}$ cannot be zero because we do not allow skipping base layers. Moreover, the deadline of any chunk $i$ is a function of the total stall from the start of the playback until the playback time of chunk $i$, \ie $deadline(i)=(i-1)L+s+d(i)$. 
We define the total stall duration from the start untill the play-time of chunk $i$ as $d(i)$. Therefore, the deadline of any chunk $i$ is $(i-1)L+s+d(i)$.  The objective function is thus given as:
\begin{equation}
\sum_{k=1}^K\sum_{n=0}^{N}\lambda_n^{(k)} G_{n}^{(k)}- \mu d(C)
 \label{equ:eq2}
 \end{equation}
where the weight for the stall duration is chosen such that $\mu \gg \lambda_{0}^1$, because users tend to prefer not running into re-buffering over playing the video in better quality. 
This is a multi-objective optimization problem with quality and stalls as the two objectives, and is formulated as follows.
\begin{eqnarray}
{\bf Maximize: } (\ref{equ:eq2})
\label{equ:eq2_ref}
\end{eqnarray}
subject to  (\ref{equ:c0eq1}), (\ref{equ:c2eq1}), (\ref{equ:c3eq1}), (\ref{equ:c4eq1}), (\ref{equ:c5eq1}), (\ref{equ:c6eq1}), (\ref{equ:c6eq1_2}), (\ref{equ:c7eq1}), (\ref{equ:c9eq1})
\vspace{0.1in}
\begin{eqnarray}
	D_{n,u}^{(k)}(i) = \sum_{j=1}^{(i-1)L+d(i)+s} z_{n,u}^{(k)}(i,j)  \forall i,  n, k, u
	\label{equ:c1eq2}
\end{eqnarray}
\begin{equation}
z_{n,u}^{(k)}(i,j)= 0\   \forall \{i: (i-1)L+s+d(i) > j\}, \forall u,k
\label{equ:c8eq2}
\end{equation}
\begin{equation}
d(i) \geq d(i-1)\geq 0\   \forall i>0 \label{equ:c9eq2}
\end{equation}
\begin{equation}
Z_{0,i} =Y_0
\label{equ:c10eq2}
\end{equation}

As compared to skip based streaming, we do not allow base layer skips as specified by Constraint \eqref{equ:c10eq2}. We note that this problem can be solved using an algorithm similar to that for the skip-based streaming. One difference as compared to the skip version is that the first step of the no-skip algorithm is to determine the minimum stall time such that all chunks are fetched at least at base layer quality since that is the first priority. Therefore, the forward scan of the base layer objective is not to find the minimum number of skips. The No-Pref-Groupcast algorithm runs a base layer forward scan that has the objective of checking if all chunks can be fetched at least at the base layer quality with the current startup delay.  At the deadline of any chunk $i$, if $V_{0,i} < i$, the algorithm increments the deadline of every chunk $ \geq i$ by $1$ and resumes the forward scan (Lines 11-15). The algorithm does not proceed to the next chunk untill the condition of $V(i) \geq i$ is satisfied. At the end, the algorithm sets the final deadline of every chunk $i$ to be: $deadline(i)=(i-1)L+s+d(C)$. Therefore, the base layer forward scan achieves the minimum stall duration and brings all stalls to the very beginning since that will offer bandwidth to more chunks and can help increase the quality of more chunks. The detailed steps are described in Algorithm 5.

The rest of the algorithm is equivalent to the skip version since skips are not allowed only for base layers, \ie higher layers can be skipped. 
The key difference in the no-skip with compare to the skip version is that the startup delay is decided such that there will be no skips. No-skip version can be considered as a special case of the skip version in which all chunks can be fetched at least at the base layer without skipping an entire chunk. The No-Pref version of No-Skip GroupCast is described in Algorithm~\ref{algo:noSkipNoPref}. 

\begin{figure}
    \scalebox{0.95}{
		\begin{minipage}{\linewidth}
			\begin{algorithm}[H]
			\small
				\begin{algorithmic}[1]
				\STATE {\bf Input:}  $\mathbf{Y}=\{Y_n\forall n\}$, $L$, $s$, $C$, $\eta_u$, $\mathbf{B}^{(u)}=\{B^{(u)}(j)\forall j\}$, $u=1,\cdots,U$. 
				\STATE {\bf Output:} $I^{(u)}_n$, $u=1,\cdots,U$: set containing the indices of the chunks that can have their $n$th layer fetched by user $u$.	
								
   \STATE  $deadline(i)=(i-1)L+s \quad \forall i$
  
    \FOR{each layer $n = 0, \cdots, N$}   
     \STATE $R^{(u)}(j)=\sum_{j^\prime=1}^{j} B^{(u)}(j^\prime)$,$\forall j,u$
	\STATE $V_{n,i}=\sum_{u=1}^U\floor{\min(\frac{\eta_u}{y_n},\frac{R^u(deadline(i))}{y_n})}$, $\forall i$
   \IF{$n=0$}
       \FOR{$i=1:C$}
       \STATE {\bf If} $i>1$ {\bf then} $d(i)=d(i-1)$,$deadline(i)=(i-1)L+s+d(i)$
       \STATE $V_{n,i}=\sum_{u=1}^U\floor{\min(\frac{\eta_u}{y_n},\frac{R^u(deadline(i))}{y_n})}$
       \WHILE{($V_{n,i} < i$)}
       		\STATE $d(i)=d(i)+1$
		\STATE $deadline(i)=(i-1)L+s+d(i)$
		\STATE $V_{n,i}=\sum_{u=1}^U\floor{\min(\frac{\eta_u}{y_n},\frac{R^u(deadline(i))}{y_n})}$
	\ENDWHILE
      \ENDFOR
	\STATE  $deadline(i)=(i-1)L+s+d(C), \quad \forall i$ , $i^\prime=1$
	\ELSE
		 \STATE $skip(0)=0$
     \FOR{$i=1:C$}
       \IF{$V_{n,i} < i-skip(i-1)$ or ($i \notin  I_{(n-1)}^{(1)}$ and,$\cdots$, $i \notin I_{(n-1)}^{(U)}$)}
      \STATE $skip(i)=skip(i-1)+1$
      \ENDIF
      \ENDFOR
        \STATE Skip the first $skip(C)$ chunks.
        \STATE $i^\prime=skip(C)+1$: the index of the first chunk to fetch
        \ENDIF

    \FOR{$i=i^\prime:C$}
    \STATE $j_u =deadline(i),u=\in \{1,\cdots,U\}$
    	\IF{($n=0$ or $i \in I_{n-1}^{(1)}$ or $\cdots$, $i \in I_{n-1}^{(U)})$}
				\STATE {\bf if}(${B}_{u} > Y_n$) ${B2}_{u}={B}_{u}$ {\bf else }  ${B2}_{u}={\bf 0}$			     \STATE $t=deadline(i-1)$
      				\STATE $[B2_{u},\zeta_i^n(u)]=$ Backward$(u,i, j_u, \mathbf{B2}_{u}, Y_n, t) \forall u$
			\STATE $u_1=\arg\min (\zeta_i^n)$
      				\STATE $I_n^{(u_1)}=I_n^{(u_1)} \cup i$, $\mathbf{B}_{(u_1)}=B2_{(u_1)}$

	\ENDIF
      \ENDFOR
     \ENDFOR
   				\end{algorithmic}
				\caption{Offline No-Skip No-Pref GroupCast Algorithm }\label{algo:noSkipNoPref}
			\end{algorithm}
		\end{minipage}
	}
	\end{figure}

%% file: simulation_main.tex
\section{System Evaluation}

\label{sec:eval}
In this section, we describe the implementation and the evaluation of skip and no-skip versions of GroupCast. However, we only report the results of the skip version because the results for both versions are qualitatively similar. 

\subsection{Evaluation Parameters}
\label{sec:evalpar}

\begin{table}[t!]
  \centering
  \caption{SVC encoding bitrates used in our evaluation}
  \begin{tabular}{|c|cccc|} \hline
    playback layer & BL & EL1 & EL2 & EL3 \\ \hline
   {Nominal Cumulative} Rate (Mbps) & 1.45 & 2.45 & 4.15 & 6.36 \\ \hline
  \end{tabular}
  \label{tab : svc_rates}
  \vspace{-.1in}
\end{table}

\begin{figure}[htbp]
	\centering
	\includegraphics[trim=0in 0in 0in 0in, clip,  width=.48\textwidth]{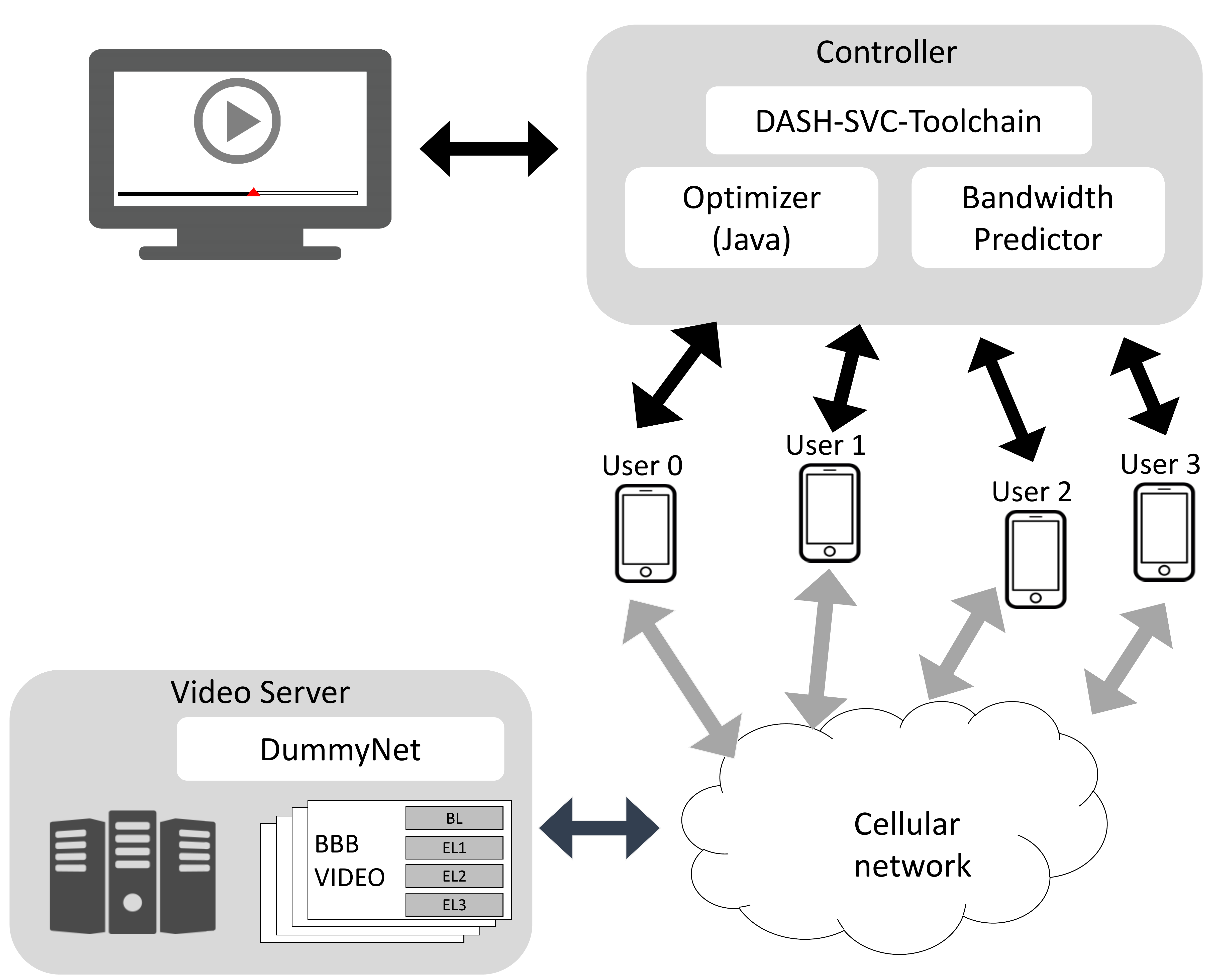}	
	\vspace{-.1in}
	
	\caption{System Setup of the Emulated Experiment.}
	\label{fig:syssetup}
	\vspace{-.1in}
\end{figure}

\begin{figure}[htbp]
	\centering
	\includegraphics[trim=0in 0in 0in 0in, clip,  width=.4\textwidth]{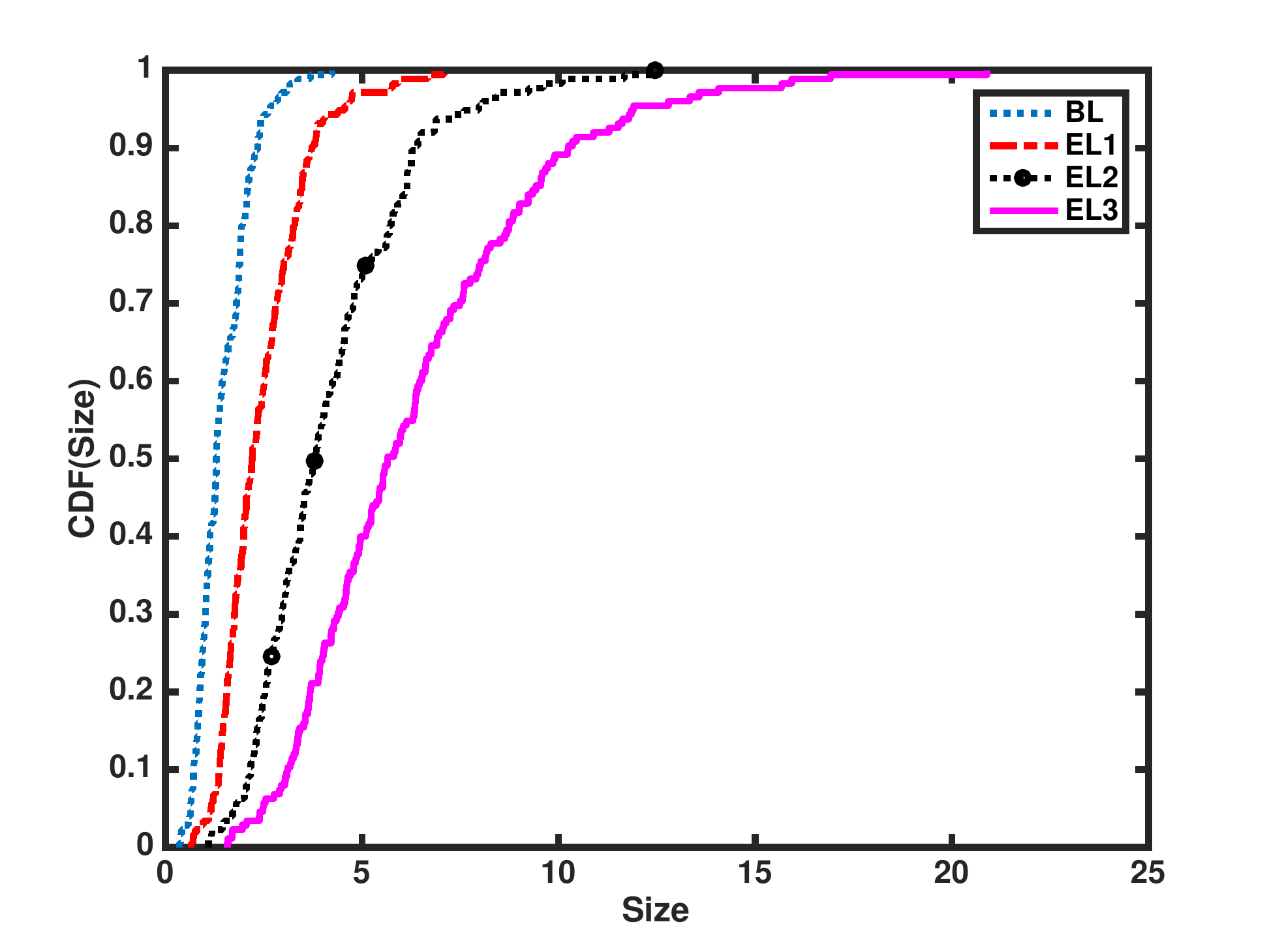}	
	\vspace{-.1in}
	
	\caption{Rate distribution of  the video used in the Implementation at different quality levels.}
	\label{fig:rate_dist}
	\vspace{-.1in}
\end{figure}

{\bf Implementation Setup.} Fig. \ref{fig:syssetup} illustrates the system setup. The implemented system contains three components: (i) the video server, (ii) the controller, and (iii) the clients. The functionalities of the video server include hosting the videos and running the bandwidth shaper. The video server runs on a Dell Power Edge R420 with a 6-core Intel E5-2620v3 CPU, 120 GB of RAM, and Ubuntu Server 16.0 LTS OS. The video server runs a Python-based HTTP Server that serves the requests of the clients. Each layer in the video is requested by HTTP GET requests using the layer's URL. We use the publicly available video ``Big Buck Bunny"~\cite{SVCDataset}. The video consists of 299 chunks, {\em i.e.}, 14315 frames, where each chunk has a duration of 2 seconds (48 frames). We use the quality of SVC scalability mode. The frame rate of this video is 24fps, and the spatial resolution is 1280x720. The video is encoded into one base layer and three enhancement layers. Table~\ref{tab : svc_rates} gives the cumulative nominal rates of each layer, where "BL" and "EL$_i$" refer to the base layer and the cumulative $i^{th}$ enhancement layer rate, respectively. For example, the exact rate of the $i^{th}$ enhancement layer is equal to EL$_i$-EL$_{(i-1)}$, with EL$_0 = $BL. The rate distribution of the different layers of the VBR video is depicted in Fig. \ref{fig:rate_dist}. For the bandwidth shaping, we use \textit{Dummynet}~\cite{dummynet} tool to limit the outgoing bandwidth from the video server according the collected bandwidth traces. All reported results are based on the 1000 bandwidth traces described next.

The controller contains three components: (i) the Optimizer, (ii) the Bandwidth Predictor, and (iii) the SVC decoder. The functionalities of the controller include providing and orchestrating the communication among users, running the optimizer, sending the decisions to the clients, receiving back the downloaded layers from clients, decoding the received chunks, and playing the video on the display. The optimizer and the bandwidth predictor are implemented in Java. We use the DASH-SVC-Toolchain to decode the video and stream the videos to the display screen~\cite{SVCDataset}. The controller can be hosted on any device, \eg a phone, a laptop, or a single-board computer, that can provide tethered or wireless connections to the clients and the display. In our implementation, we use a MacBook Pro with 2.4 GHz Intel I5 processor and 8 GB DDR3 RAM as the controller, which communicates with the clients using an 802.11n WiFi link. Such link does not incur communication overhead, \ie latency, since the clients are within the proximity of the controller. The 802.11n network can operate with a net data rate of at least 100 Mbit/s \cite{xiao2005ieee}, which is sufficient to serve the users in our setup without being a bottleneck. 

The functionalities of the Clients include receiving the fetching policy from the controller, sending the HTTP GET request for each assigned layer to the Video server, and streaming the downloaded layers to the controller. Note that the clients do not decode the downloaded SVC chunks. The client application is implemented in Java and deployed as a mobile phone application. We use four phones as clients, each running Android 7.0. Each phone has 4G of RAM and 32GB of internal memory. We assume that the only connection to the Internet for the controller is provided by a cellular network (LTE/4G/3G) via the clients.  

{\bf Bandwidth Traces.} We use a public dataset that consists of continuous one-second measurement of throughput for cellular system \cite{riiser2013commute}. The dataset has been divided into 1000 traces, each of six minutes length \cite{MPC}. The statistics of the dataset are shown in Fig.~\ref{fig : bwStatV2}(a-b). We assign $250$ bandwidth traces to each user at random. The average throughput across the traces varies from 0.7 Mbps to 2.7 Mbps, with a median of 1.6 Mbps. In each trace, the instantaneous throughput is also highly variable, with the average standard deviation across traces being 0.9 Mbps. We only consider the first 175 chunks of the video because that the video length is longer than the provided bandwidth traces.

\begin{table}[htb]
	  \vspace{-.1in}
  \centering
  \caption{Max contributions used in our evaluation}
  \begin{tabular}{|c|cccc|} \hline
    User No.& 1 & 2 & 3 & 4 \\ \hline
   {Max Contribution} (Mb) & 672 & 504 & 336 & 168 \\ \hline
  \end{tabular}
  \label{tab : max_cont}
  \vspace{-.1in}
\end{table}

{\bf Experiment Parameters.}  We considered three experiments scenarios. The first scenario is the \textit{no-preference and infinite contribution}, \ie all users contribute equally likely while imposing no maximum contribution to any of the users. Second scenario is the \textit{no-preference and finite contribution}, which imposes maximum contribution for all the users. Table \ref{tab : max_cont} gives the maximum contribution of the four users in $Mb$ that are used in the evaluation. We note that the sum of the maximum contributions is 1400 Mb, which allows to fetch most of the chunks at the second enhancement layer given that the bandwidth is enough. Notice that the expected size of all the chunks at the second enhancement layer is $4.8\times 2\times 175 = 1680$ $Mb$. The third scenario is the \textit{preference-aware and finite contribution} in which users impose the maximum contributions specified in Table \ref{tab : max_cont}. Moreover, users 3 and 4 are less preferable in the sense they can only help in avoiding skips, \ie  fetch base layers if the other two users cannot meet the deadlines.

We assume a playback buffer of 10 seconds for all the scenarios considered in the evaluation. In other words, the window size is $W=5$ chunks. Moreover, for all online algorithms, we re-consider the decisions 4 four seconds, \ie $\alpha=4$ $s$, and the parameter $\delta$ is also assumed to be 2 seconds. Finally, the startup delay is 5 seconds.


\begin{figure}
\input{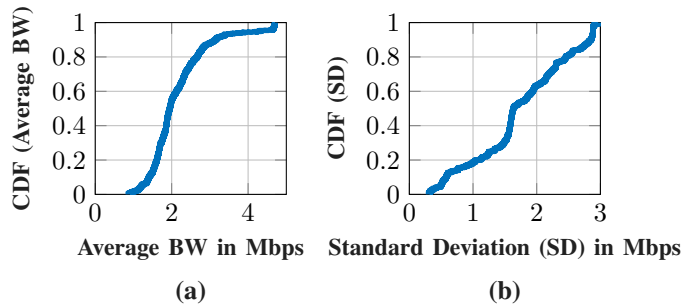}
 \vspace{-.3in}
 \caption{Statistics of the two bandwidth traces: (a) mean, and (b) standard deviation of each trace's available bandwidth.}
 \label{fig : bwStatV2}
 \vspace{-.1in}
 \end{figure}

{\bf Bandwidth Prediction. } Although our algorithm can work with any bandwidth prediction technique, we consider a harmonic mean-based prediction. The harmonic mean of the throughput achieved in fetching the last 5 layers by the $i^{th}$ user is used as a predictor for the bandwidth of that user for the next window of chunks. Since there is no prediction at the beginning, the first 4 chunks will be assigned randomly to the 4 users at base layer quality. Consequently, each user will be fetching one of the first 4 chunks. 


{\bf Comparison Baselines. } The proposed online algorithm is compared with the following two strategies based on round robin assignment of layers to users. To achieve fair comparison with proposed algorithms, we consider maximum contribution per window in the round robin assignment. In particular, a layer is skipped by the user if his residual contribution limit in the window is not enough to fetch that layer. We also include the offline algorithm to show the relative performance of the compared algorithms to our offline scenario, which has perfect knowledge of the bandwidth for the whole period of the video.


{\bf Buffer-based Cooperative streaming Approach (BB).} BB adjusts the streaming quality based on the playback buffer occupancy. Specifically, the quality follow the same strategy of BBA algorithm \cite{BBA} in making the quality decision. The decision depends on two conditions. First, if the buffer occupancy is lower (higher) than the lower (higher) threshold, then chunks are fetched at the lowest (highest) quality. Second, if the buffer occupancy lies in between the two thresholds, then the buffer rate relationship is determined by a linear function between the two thresholds. We use $4$ s and $10$ s as the lower and higher thresholds on the buffer length, respectively. 

Once the decision is obtained, the different layers of the next window of chunks are assigned to users 1, 2, 3, and 4 using round robin strategy for the no-preference scenario. If the maximum contribution of the current window for any of the users does not allow to fetch the base layers, then the next users in the round is chosen. The preference scenario works the same except that if the decided quality is higher than the base layer quality and the maximum contribution of the preferred links allow to fetch base layers, then the enhancement layers are not assigned to the less preferred users. As a result, the less preferred users do not fetch beyond the base layer. If fetching any layer causes to violate the maximum contribution of this window for all users, then that layer is skipped. The process is repeated every $\alpha$ seconds to decide the fetching policy of the next $W$ chunks similar to GroupCast algorithm.

{\bf Prediction based Cooperative streaming Approach (PB).} 
PB uses the harmonic mean to predict the bandwidth for the next window of chunks. In particular, PB computes the sum of the predicted bandwidths for all users for the no-preference scenario, and for users 1 and 2 for the preference scenario. Then, it considers $90\%$ of the computed value as the predicted bandwidth. The closest quality level that is less than the value of the predicted bandwidth is used as the quality decision of all chunks of this window. The decided layers to be fetched are distributed in round robin strategy as described for the BB algorithm.


\input{simulation_a}



%% file: simulation_a.tex
\begin{table}[htb]
	  \vspace{-.1in}
  \centering
  \caption{Skip Percentage and Average Playback rate comparison}
  \begin{tabular}{|c|} \hline
  No Preference, Infinite Contribution\\ \hline
   \end{tabular}
  \begin{tabular}{|c|cccc|} \hline
    Algorithm& BB & PB & GroupCast & off-GroupCast \\ \hline   
    \% of Skips&4.89 &3.5 & 0.28&0\\ \hline
    APBR(Mbps)&5.77&4.0&6.18&6.35\\ \hline
     \end{tabular}
 \begin{tabular}{|c|} \hline
  No Preference, Finite Contribution\\ \hline
   \end{tabular}
    \begin{tabular}{|c|cccc|} \hline
    Algorithm& BB & PB & GroupCast & off-GroupCast \\ \hline   
    \% of Skips&6.21 &4.89 & 2.08&0\\ \hline
    APBR(Mbps)&4.20&3.90&4.46&4.80\\ \hline
  \end{tabular}
   \begin{tabular}{|c|} \hline
  Preference-aware, Finite Contribution\\ \hline
   \end{tabular}
    \begin{tabular}{|c|cccc|} \hline
    Algorithm& BB & PB & GroupCast & off-GroupCast \\ \hline   
    \% of Skips&16.75 &11.68 & 1.80&0\\ \hline
    APBR(Mbps)&2.75&2.00&3.37&3.35\\ \hline
  \end{tabular}

  \label{tab: compTab}
  \vspace{-.1in}
\end{table}

\begin{figure*}
\centering
\includegraphics[trim=0in 0in 0in 0in, clip,  width=\textwidth]{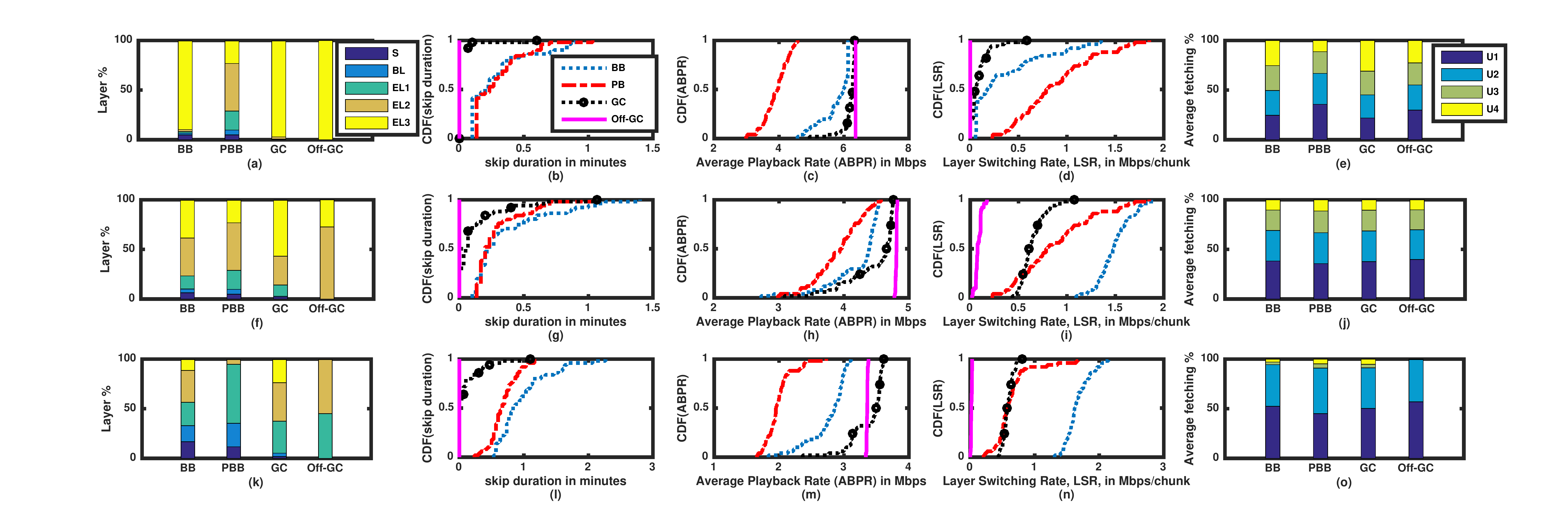}		
 \vspace{-.3in}
 \caption{ Comparing streaming algorithms: (a,b,c,d,e) are Layer percentage, CDF of the skip duration, CDF of the Average Playback Rate, APBR, CDF of the Layer Switching Rate, and  fetching percentage per user for the no-preference infinite contribution scenario, (f,g,h,i,j) and (k,l,m,n,o) show the same thing for no-preference finite contribution and preference-aware scenario, respectively.}
 \label{fig : coopFig}
\end{figure*}

\subsection{Evaluation Results}
\label{sec:eval_skip}

In this subsection, we compare the performance of the GroupCast algorithm 
with the baseline algorithms described in Section~\ref{sec:evalpar}. We denote the offline fetching policy by \textit{Off-GroupCast}, where the bandwidth is predicted perfectly and the decisions of all chunks are calculated in one run of the optimization problem. The maximum contributions are constrained as given in Table \ref{tab : max_cont}.

\begin{thisnote}
The results are illustrated in Table~\ref{tab: compTab} and Fig.~\ref{fig : coopFig}. From the results, Off-GroupCast can fetch all chunks all chunks at least at the base layer quality without running into skips in all three scenarios. Moreover, we observe that GroupCast achieves the minimum number of skips and the highest playback average among all online algorithms. GroupCast significantly outperforms BB and PB. For example, over all the traces, GroupCast runs in only $0.28\%, 2.08\%$, and $1.80\%$ of skips in the no-preference infinite contribution, no-preference finite contribution, and preference-aware finite contribution, respectively. On the other hand, BB runs into $4.89\%, 6.21\%$, and $16.75\%$ skips, while PB runs into $3.1\%, 4.89\%$, and $11.68\%$ skips for same respective scenarios.

Fig.~\ref{fig : coopFig}-(a, f, k) show the layer percentage for each scenario. Fig. \ref{fig : coopFig}-a illustrates that there is an assignment policy such that all chunks can be fetched at the highest quality. Moreover, we see that even though GroupCast uses short prediction window with prediction errors, it still achieves the closest performance to the offline algorithm which has the perfect knowledge of the future bandwidth. This is because that we use the sliding based window in which decisions are reconsidered every $\alpha$ seconds.

Fig.~\ref{fig : coopFig}-(f) demonstrates that the quality drops when the maximum contribution constraint is imposed. According to the maximum contribution values imposed in Table~\ref{tab: compTab}, we expect that most of the chunks will be fetched at the $2^{nd}$ enhancement layer quality, which is also reflected by the off-GroupCast algorithm in Fig.~\ref{fig : coopFig}-f. We also note that GroupCast is the best algorithm to adjust to the imposed maximum contribution constraints among online algorithms. However, since the online algorithm  only has local information per window every $\alpha$ seconds, we observe that it runs into fetching more chunks at the highest quality at the cost of running into subsequent skips.

Fig.~\ref{fig : coopFig}-(k) shows the results of preference-aware with finite contribution in which users 3 and 4 can only help in avoiding skips. The figure shows that Off-GroupCast can ideally fetch all chunks at either E1 or E2 quality levels. We note similar observations in the no-preference finite contribution scenario, \ie GroupCast significantly outperform the round robin strategies in avoiding skips and fetch more chunks at higher quality. Moreover, we observe that because of the short prediction of the online GroupCast, it fetches more chunks at the highest quality levels at the cost of skipping some other chunks. Therefore, GroupCat achieves higher average than Off-GroupCast at the cost of more skips, which further degrades the QoE according to our formulation.


Fig.~\ref{fig : coopFig}-(b, g, f) and Fig.~\ref{fig : coopFig}-(c, h, m) show the CDF of the skip percentage and the average playback rates of all of the algorithms over all of the bandwidth traces, respectively. We clearly see that GroupCast achieves the minimum number of skips and the highest average playback rate almost in every single bandwidth trace which reflects the adaptability of GroupCast to different bandwidth regimes and oscillations.

Fig.~\ref{fig : coopFig}-(d) plots the distribution of the layer switching rate (LSR) for the all bandwidth traces. The LSR of a video is defined as $\frac{1}{C}\sum_{i=2}^C |X(i)-X(i-1)|\cdot {\bf 1}(\Gamma(i) \neq \Gamma(i-1))$, where $C$ is the number of chunks and $\Gamma(i)$ is the highest layer fetched for the $i^{th}$ chunk. Thus, if every two neighboring chunks are fetched at the same layer, then the layer switching is zero even if the layers have different sizes. In other words, we only account for size difference in jumping from one layer to another. As illustrated by the figure, GroupCast achieves a significantly lower LSR as compared to BB and PB. We note that the LSR is lower than 1 Mbps with probability 1 for GroupCast and Off-GroupCast. However, GroupCast has higher LSR than Off-GroupCast because GroupCast has erroneous bandwidth prediction and only makes local window-based decision.

Fig.~\ref{fig : coopFig}-(e) shows the percentage of the content fetched by each user. According to the settings of preference-aware finite contribution scenario, users 3 and 4 should only be used to avoid skips. Off-GroupCast shows that users 3 and 4 are used to fetch a negligible number of base layers (almost 0). However, we observe that GroupCast uses users 3 and 4 more than required. The reason is due to the local window-based decision and the bandwidth prediction error of GroupCast. As observed in the figure, GroupCast goes up to fetching chunks at E3 quality level and, as a result, affects the quality of subsequent chunks that cannot meet their deadlines by using only the first two users. Nevertheless, GroupCast uses users 3 and 4 more efficiently than BB and PB to avoid further skips because that GroupCast incorporates both users 3 and 4 with their imposed constraints in the optimization-based decision.

\end{thisnote}
In conclusion, GroupCast algorithm is able to distribute the layers among users efficiently, thus achieving significantly higher QoE as compared to the considered baselines. We observe that incorporating the chunk deadlines into the optimization problem and favoring the subsequent chunks is essential for the success of the algorithm. This is a unique feature in GroupCast that allows to achieve both low skip duration and high playback quality, while favoring lower layers over the higher layers.


%% file: Conclusion.tex
\section{Conclusions}\label{sec:concl}

In this paper, we propose a preference-aware cooperative video streaming algorithm for SVC-encoded videos. We consider both skip and no-skip based streaming. Finding the quality decisions of the video's chunks and the fetching policy of the SVC layers subject to the available bandwidth, chunk deadlines, and cooperation willingness of the different users was formulated as non-convex optimization problem. A novel algorithm was developed to solve the proposed optimization problem. The proposed algorithm has a polynomial run-time complexity. Real implementation on android devices using SVC-encoded video and bandwidth traces from a public dataset reveal the robustness and performance of the proposed algorithm. The results show that the proposed algorithm improves the number of skips/stalls by at least 57\% as compared to the considered baselines while also improving the average quality.

The paper considers streaming using multiple clients, when the video is encoded with a layered code. Extending the algorithm and implementation over MPEG-DASH with byte-range requests is an interesting future direction. 

%% file: examples.tex
\section{Illustration Examples for Pref/No-Pref GroupCast} \label{examples}

\subsection{Example 1: No-Pref GroupCast}
\label{ex1}
\begin{figure*}[htbp]
\centering
\includegraphics[trim=0.6in 1.5in 0.6in 1.5in, clip,  width=\textwidth]{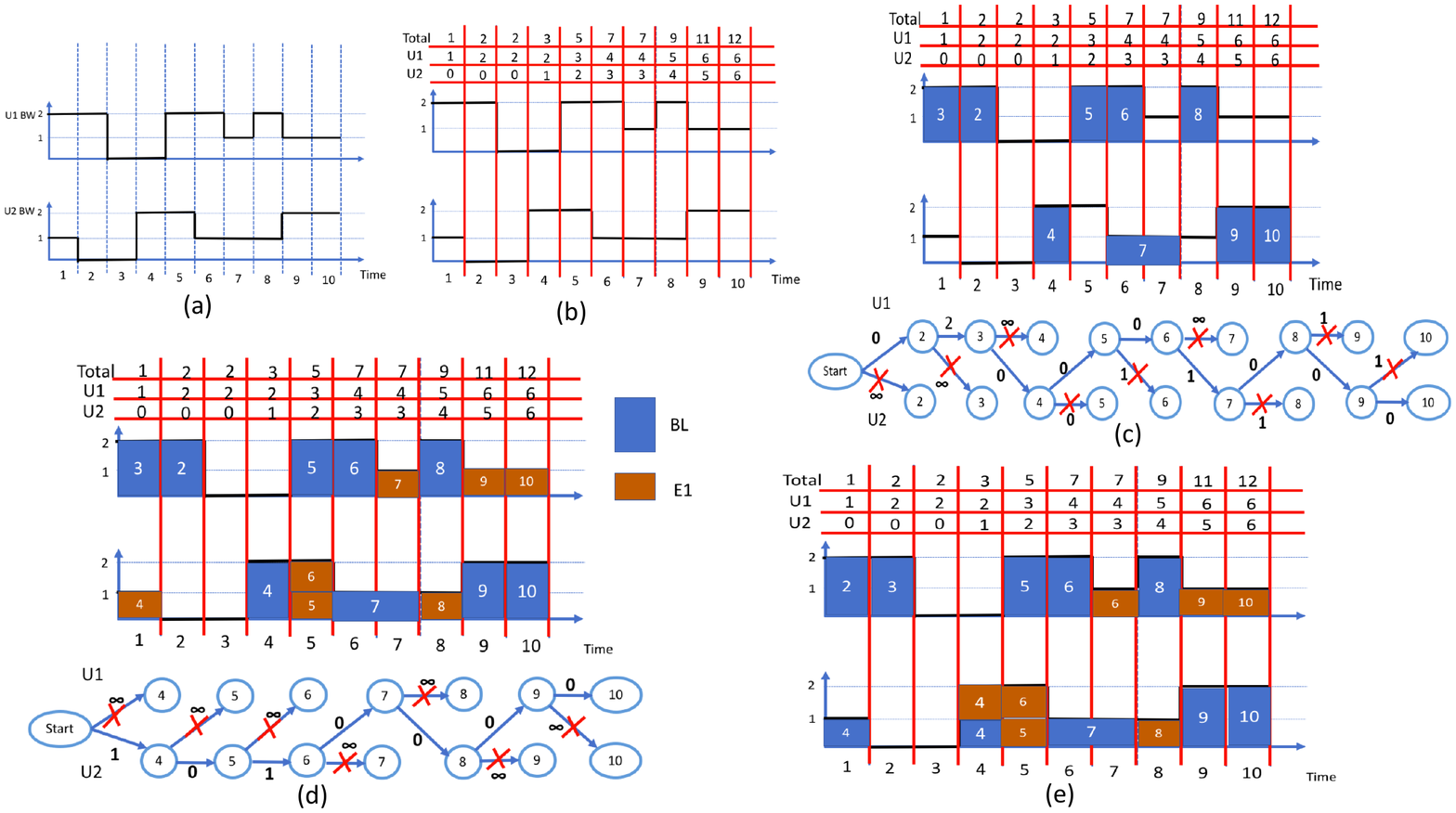}	
 \vspace{-.3in}
 \caption{ Illustration Example of No-Pref GroupCast Algorithm: (a) Bandwidth Trace, (b) Base Layers that can be fetched before the deadline of every chunk, (c) Base Layer Decision, (d) $1^{st}$ Enhancement Layer Decision, and (e) The Actual Fetching Policy}
 \label{fig:ex1}
\end{figure*}

 Fig.~\ref{fig:ex1} shows an example that illustrates how No-Pref GroupCast algorithm works. We assume a video that consists of 10 chunks, 1 second length each. The video is an SVC encoded into 1 Base Layer ($BL$) and 1 Enhancement Layer ($E_1$). The $BL$ and $E_1$ sizes are $2Mb$, and $1Mb$ respectively. \ie  $Y_0=2Mb$, and the $Y_1=1Mb$. Moreover, we assume that the startup delay is 1 second. Therefore, the $deadline(i)=i, \forall i$. 

We assume a scenario of two users U1 and U2. Fig.~\ref{fig:ex1}-a shows the bandwidth traces of the two users, and Fig.~\ref{fig:ex1}-b show the result of the first forward scan. The forward algorithm finds the maximum number of base layers that can be fetched before the deadline of every chunk. We clearly see that up to the deadline of the $3$rd chunk, only 2 chunks can be fetched. Therefore, one out of the first 3 chunks should be skipped. The algorithm as explained previously decides to skip the first chunk (chunk 1) since the bandwidth that the first chunk leaves can be available to all of the remaining chunks for the next layer decisions. In other words, if we skip chunk 3, then it is possible that part or the whole bandwidth that it leaves comes after the deadline of chunks 1 and 2 which means that chunks 1 and 2 can't benefit from this bandwidth in fetching their higher layers.

Forward scan finds the chunks that can have their base layers fetched without violating their deadlines. Consequently, Backward scan described in Fig.~\ref{fig:ex1}-c finds the fetching policy such that all these chunks have their base layers fetched without violating their deadlines as promised by forward algorithm. Moreover, backward algorithm finds the policy that maximizes the total bandwidth of every chunk for the next layer decision. 

 As shown in Fig.~\ref{fig:ex1}-c, the algorithm simulates fetching the base layer of every chunk starting from its deadline going backwards. Moreover, it does not consider chunk $1$ since forward algorithm decided that chunk 1 can't be fetched. The algorithm proceeds in the order of the chunks, finds $\zeta_{i,0}(u), \forall u$, and decides which link should fetch each chunk. The bottom subplot of Fig.~\ref{fig:ex1}-c shows how the assignment decision of chunks is made. First step is to calculate the cost of fetching chunk 2 by each of the users (U1 and U2). The cost of fetching chunk $2$ by the first user is $0$ ($\zeta_{2,0}(1)=0$) since fetching the chunk starting from its deadline by user 1 will lead to fully downloading it before crossing the deadline of chunk 1. However, chunk 2 can't be fetched by the second user, so $\zeta_{2,0}(2)=\infty$. The costs are marked on the arrows that points from the decision of a chunk to the next one in the bottom subplot of Fig.~\ref{fig:ex1}-c. Following the same strategy, we see that chunks 2, 3, 5, 6, and 8 are assigned to the first user, and chunks 4, 7, 9, and 10 are assigned to the second user. One more thing worths mentioning, we see in the top subplot of Fig.~\ref{fig:ex1}-c that chunk 3 is fetched before 2. To explain this, there are two things to point out here. First, this decision is not the final fetching policy, it just provides the chunk assignment policy such that the number of the current layer skips is minimized and the total bandwidth available before the deadline of every chunk for the next layer decision is maximized . Second, since we simulate fetching chunks backward starting from their deadlines and in the order of the chunks, and chunk number 2 was fetched first. Hence, the only available bandwidth to fetch chunk 3 is at the first time slot which led to considering fetching chunk 3 before chunk 2. We will see in the final decision, Fig.~\ref{fig:ex1}-e which represents the actual fetching policy, the chunks 2 and 3 will be fetched in their order, so the current decision of chunks 2 and 3 will be reversed.
 
 
 Fig.~\ref{fig:ex1}-d shows how the algorithm repeats the same process described in Fig.~\ref{fig:ex1}-c but for the $1^{st}$ enhancement layer decisions. In the $E1$ decisions, the algorithm uses the remaining bandwidth of each link after excluding the amount reserved for fetching the $BL$s. Moreover, since the $BL$ of the first chunk is not fetched, its $E1$ is not considered. Finally, Fig.~\ref{fig:ex1}-e shows the actual fetching of the chunks. As shown in the figure, each link fetches the layers in order of the chunks they belong to, and for the same chunk, the layers are fetched according to their order. For example,  $E1$ of the $4$-th chunk is fetched before $E1$ of the $6$-th chunk.

\subsection{Example 2: Pref GroupCast}
\label{ex2}

\begin{figure*}[htbp]
\centering
\includegraphics[trim=0.6in 1.5in 0.4in 1.5in, clip,  width=\textwidth]{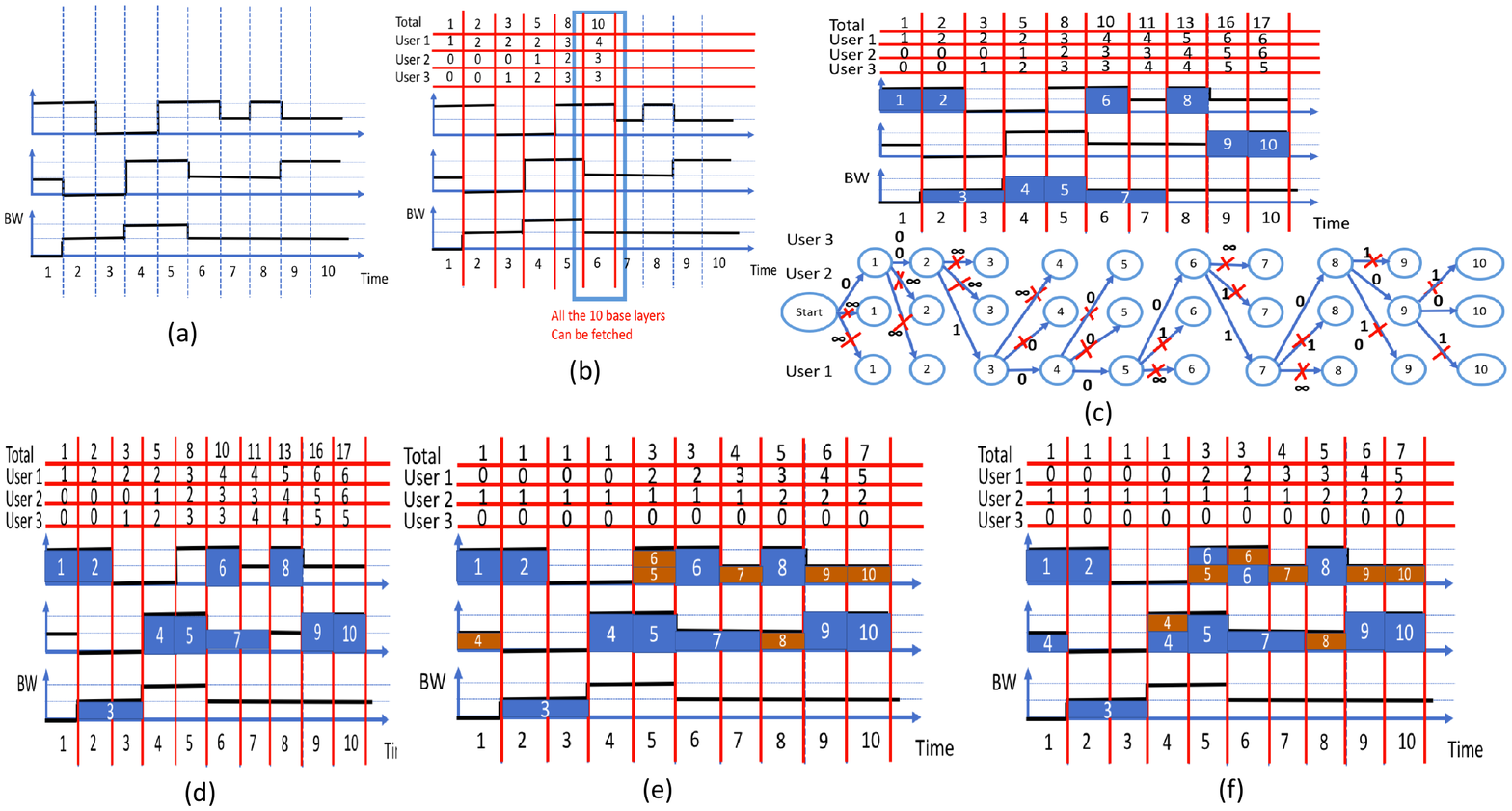}	
 \vspace{-.3in}
 \caption{ Illustration Example of Pref  GroupCast Algorithm: (a) The bandwidth traces, (b) The Skip decision, (c) After the first run of No-Pref GroupCast, (d) After the second run of No-Pref GroupCast, (e) Enhancement layer decisions, (f) The Final Fetching Policy}
\label{fig:ex2}
\end{figure*}

This example illustrates how Pref GroupCast works. We assume that user 3 is used only to avoid skips, and user 2 can help fetching up to the first enhancement layer if user 1 fails to do that. Bandwidth traces are shown in  Fig.~\ref{fig:ex2}-a, and the video parameters are as those described in example 1. Fig.~\ref{fig:ex2}-c shows the initial fetching policy after the first call of No-Pref GroupCast. According to the initial call of No-Pref GroupCast, all of the chunks can be fetched at least at base layer quality.

Fig.~\ref{fig:ex2}-d shows the outcome of the second call of No-Pref GroupCast. In this call only links 1 and 2 and chunks that were initially decided to be fetched by link 3 ($3,4,5$, and $7$) are considered. As a result to the second call of No-Pref GroupCast, chunks 4, 5, and 7 are moved to link 2. Running exchange algorithm will not change any decisions made already since chunk 3 is the earliest chunk that can be fetched by the 3rd user. Fig.~\ref{fig:ex2}-e shows the $E1$ decisions. Note that user 3 doesn't fetch enhancement layers. Finally, Fig.~\ref{fig:ex2}-f shows the actual sequence of chunk downloads.

%% file: skipproofs.tex

\section{Proof of Theorem \ref{them:skip1}}
\label{apdx_skip1}
The forward scan for base layers decides to skip a base layer of a chunk $i$ only if the total bandwidth up to a deadline of a chunk $j \geq i$ is not enough to fetch all chunks $1$ to $j$. Since, the bandwidth up to the deadline of the $j^{th}$ chunk is not enough to fetch all chunks $1$ to $j$. Any other feasible algorithm, \ie algorithm that does not violate the bandwidth constraint, must have a skip for a chunk $i^\prime \leq j$. Thus, for every base layer skip of the proposed algorithm, there must be a base layer skip or more for any other feasible algorithm. The forward scan repeats for every enhancement layer in order with the remaining bandwidth. For every layer $n$, the proposed algorithm decides to skip the $n^{th}$ layer of a chunk $i$ in two cases:
\begin{itemize}
\item If the $(n-1)^{th}$ layer of this chunk is not decided to be fetched. Thus, any feasible algorithm has to skip the $n^{th}$ layer of a chunk if the $n^{th}$ layer is not fetched; otherwise constraint~(\ref{equ:c3eq1}) will be violated.

\item The remaining bandwidth after excluding whatever reserved to fetch layers $1$ to $n-1$ for all chunks is not enough to fetch the $n^{th}$ layer of a chunk $j \geq i$. Thus, any feasible algorithm will decide to skip an $n^{th}$ layer of a chunk $i^\prime < j$. Otherwise the bandwidth constraint will be violated
\end{itemize}

Therefore, for any $n^{th}$ layer skip of the proposed algorithm, there must be an $n^{th}$ layer skip or more for any feasible algorithm. Thus, the algorithm achieves the minimum number of skips and that concludes the proof.

%% file: symbol.tex
\begin{thisnote}
	\section{Table of Symbols}
\label{app:symbols}

Table \ref{tbl:sym} describes the table of key notations that are used in this paper. 


\begin{center}
\begin{table}[h]
	{
  \caption{Table of Symbols}
  \begin{tabular}{ |c|c| }
    \hline
Symbol & Description \\ \hline
U & Number of users \\ \hline 
N & Number of layers \\ \hline 
C & Number of chunks \\ \hline 
L & Chunk's length \\ \hline
K & Number of users sets \\ \hline
BL & Base layer \\ \hline 
$s$ & Setup delay \\ \hline
$E_i$ & $i^{th}$ enhancement layer \\ \hline 
$Y_n$ & Size of the $n^{th}$ layer \\ \hline 
$r_0$ & Base layer rate \\ \hline
$r_i$ & $i^{th}$ enhancement layer rate \\ \hline 
$U_k$ & Users of the $k^{th}$ set \\ \hline 
$N_k$ & Maximum layer fetched by users of the $k^{th}$ set \\ \hline 
$\eta^k_u$ & Maximum download contribution of user $u$ in set $k$  \\ \hline 
$Z_{n,i}$ & Amount of $n^{th}$ layer in chunk $i$ that is downloaded  \\ \hline 
$z_{n,u}^{k}(i,j)$ & Amount of $n^{th}$ layer in chunk $i$  fetched in time slot $j$  \\ \hline 
${\bf L}_{k,u}$ & Link of user $u$ in set $k$  \\ \hline 
$D_{n,u}^{k}(i)$ & Total fetched amount of  $n^{th}$ layer of chunk $i$   \\ \hline 
$G_{n}^{k}(i)$ & Total fetched amount of  $n^{th}$ layer by users in set $k$   \\ \hline 
$B_{u}^{k}(j)$ & Available bandwidth of  user $u$ in set $k$  \\ \hline 
  \end{tabular}
\label{tbl:sym}}
\end{table}
\end{center}
\end{thisnote}

%% file: sens_anal.tex
\begin{thisnote}
\section{Sensitivity Analysis }
\label{sec:sys_analysis}
In this section, we systematically study the impact of various parameters including the prediction window size, the algorithm update frequency,  the number of collaborating users, and the chunk duration. In all experiments, we use the same bandwidth traces. In each experiment, we vary a parameter while fixing all other parameters for four users. We set the maximum contribution of each user according to Table \ref{tab : max_cont}. 


\begin{figure*}[ht]
\includegraphics[trim=0in 0.1in 0in 0in, clip, width=\textwidth]{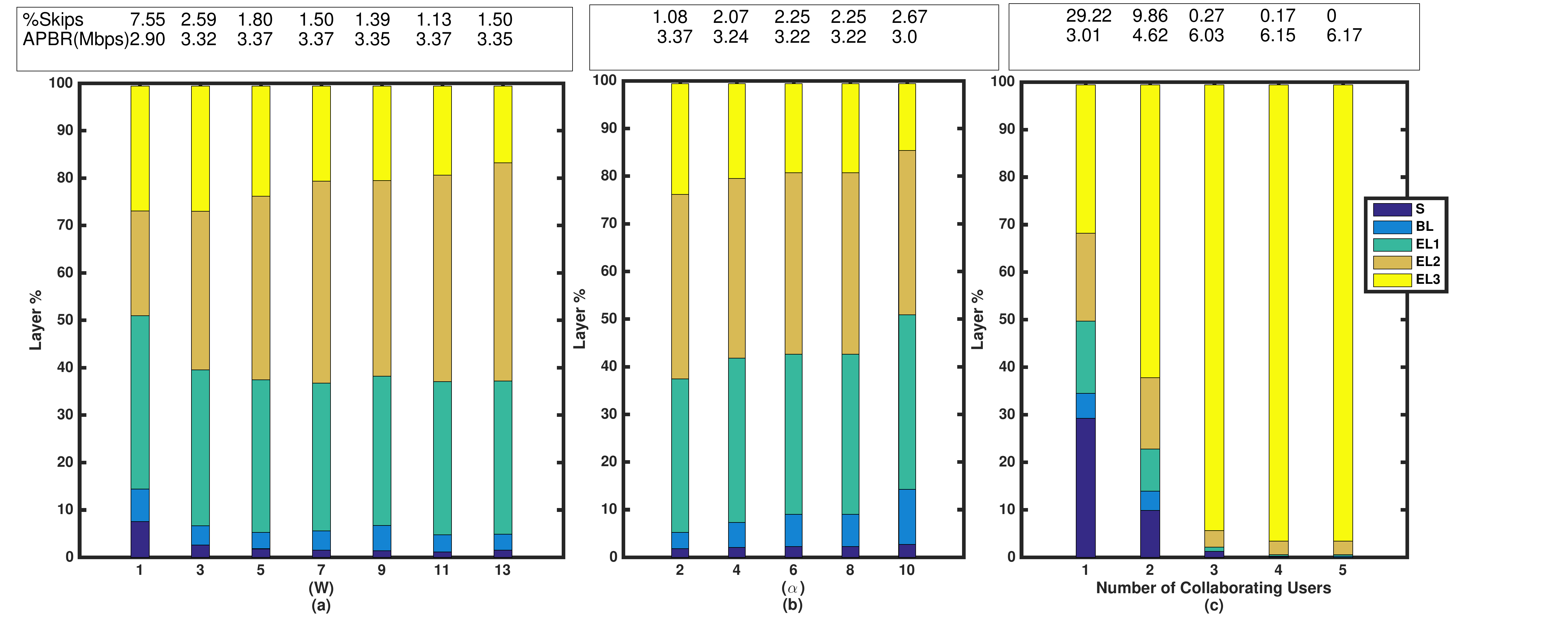}	
  \vspace{-.3in}
 \caption{Sensitivity Analysis of GroupCast, (a) Impact of the prediction window size $W$, (b) Impact of the update frequency $\alpha$, and (c) Impact of the number of contributing users}
  \label{fig : wPlot}
  \end{figure*}

{\bf Impact of prediction window size $W$.}
In the experiment, we vary the window size $W$ and plot the layers breakdown. The results are shown in Fig.~\ref{fig : wPlot}-(a). We choose $W=1$ to $13$ chunks with a step size of $2$. From the figure, we observe that our algorithm suffers from lower quality and high number of skips for small values of $W$, e.g., $W=1$. This is because that clients aggressively fetch earlier chunks at EL3 as the the algorithm does not leverage enough future bandwidth information. On the other hand, Fig.~\ref{fig : wPlot}-(a) shows that by increasing $W$ beyond $11$ chunks, we start to see an increase in the number of skips and reduction in the average playback rate since the predicted bandwidth will not be accurate for long time ahead. 

{\bf Impact of the algorithm update frequency ($\alpha$).} Recall that in the online scheme, GroupCast computes the decision every $\alpha$ seconds in order to adapt to the bandwidth prediction, which is within a limited time window and/or subject to be constantly updated.
To study the impact of $\alpha$, we run the online skip based algorithm with $W=5$ and 
varying $\alpha$ from $2$ to $10$ seconds with a steps of $2$.
%
Fig.~\ref{fig : wPlot}-(b) shows the layers breakdown for different $\alpha$ values.
As illustrated by the figure, with smaller $\alpha$, our algorithm can rerun the optimization more frequently and hence be more adjustable to bandwidth prediction errors. As a result, our algorithm obtain better results, as indicated by less skips. Moreover, we are able to use a very small $\alpha$ to recompute the scheduling frequently with small runtime overhead because of the low complexity of the scheduling algorithm. 

  {\bf Impact of the number of collaborating users.} To study the impact of the number of collaborating users, we fix $W$ to be $5$ chunks and $\alpha$ to be $2$ seconds, and vary the number of collaborating users from $1$ to $5$ users. We assume no preference among users and no maximum contribution imposed constraints. We stopped at 5 users because we observed that that all chunks can be fetched at their highest quality levels with no skips with more than 5 users.
%
Fig.~\ref{fig : wPlot}-(c) shows the layer breakdown for different number of users.
As shown, with only one user, about $1/3$ of the chunks were skipped. As the number of collaborating users increases, the algorithm starts to manage quality decisions and the layer assignment to users such that the number of skips is minimized. We observe a significant reduction in the number of skips even with only two cooperating users, while the skips are totally eliminated with 5 collaborating users.

{\bf Impact of the chunk duration.} To evaluate the impact of the chunk duration on the performance of the algorithm, we have used synthetic rates of 1.5, 2.75, 4.8, and 7.8 Mbps to represent the cumulative nominal rates of layers from BL up to the EL3. We test chunk durations of 1, 2, 3, and 4 seconds. Therefore, the corresponding video length in chunks to each of the chunk duration  is 360, 180, 120, and 90 chunks, respectively. We have carefully chosen the bandwidth traces such that we do not run into skips in any of the traces. The results is listed in Table \ref{tab : segLen}.

 As observed from the table, the average playback decreases as the chunk size grows. This is because that we make a decision for longer playback time. Therefore, every playback second of that chunk needs to be delivered at the same quality. Note that the decision is not adjusted in the middle of the chunk to any bandwidth changes. Thus, shorter chunks allow for making more fine-grained decision, \ie making a decision per 1 second of playback allows for delivering every second at the best quality. In contrast, for the case when the chunk duration is 4 seconds, all of the 4 seconds need to be delivered at the same quality. Moreover, the chunk duration of 1 second allows a faster adjustability to network changes than making a decision per 4 seconds of playback duration.

\begin{table}[t!]
  \centering
  \caption{Impact of the chunk duration}
  \begin{tabular}{|c|cccc|} \hline
    Chunk length in seconds & 1 & 2 & 3 & 4 \\ \hline
   Average Playback rate in Mbps & 7.6 & 7.3 & 7.0 & 6.8 \\ \hline
  \end{tabular}
  \label{tab : segLen}
  \vspace{-.1in}
\end{table}

 {\bf Example Illustration.} We now study the preference-aware with finite contribution scenario in more detail. In Fig.~\ref{fig:coopFig2}, we plot an example of the bandwidth traces of users 1, 2, 3, and 4, while in Fig.~\ref{fig:coopFig3} we plot the quality decisions of the four algorithms. We plot ``Max layer ID +1" with respect to time, where ``Max layer ID" is defined in Table \ref{tab : maxL}. In the preference-aware scenario, users 1 and 2 can contribute to fetch all the chunks, while users 3 and 4 can only contribute to avoid skips. The total video length is 350 seconds, \ie 175 chunks $\times$ 2 seconds. 
 
 \begin{figure}[t!]
 	\centering
 	\includegraphics[trim=0in 0in 0in 0in, clip,  width=0.48\textwidth]{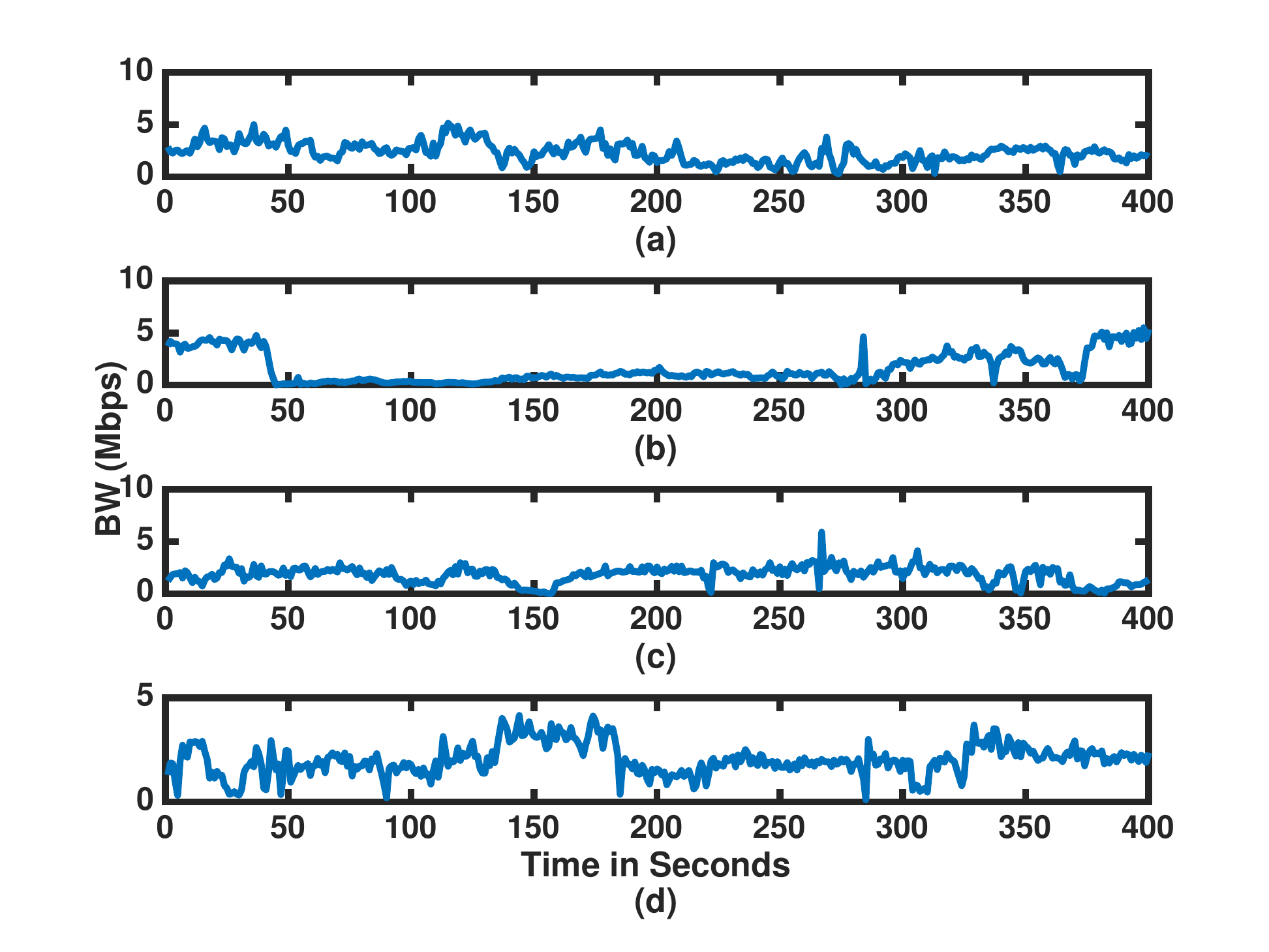} 
 	\caption{ Bandwidth traces for 4 users: (a) User 1, (b) User 2, (c) User 3, (d) User 4}
 	\label{fig:coopFig2}
 \end{figure}
 \begin{figure}[t!]
 	\centering
 	\includegraphics[trim=0in 0in 0in 0in, clip,  width=0.48\textwidth]{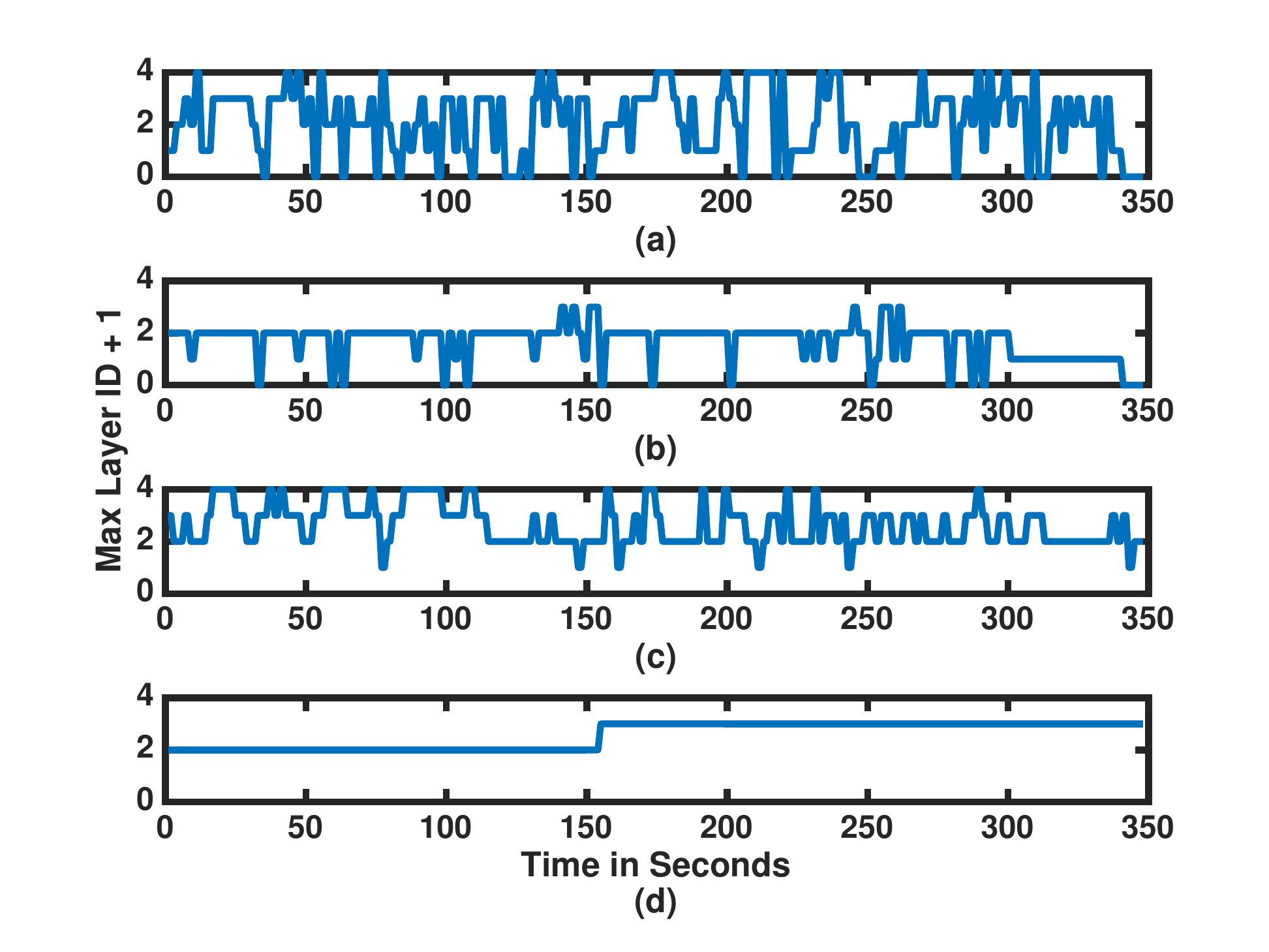} 
 	\caption{ Playback Quality Comparison of the 4 Algorithms: (a) BB, (b) PB, (c) GroupCast, (d) Offline GroupCast}
 	\label{fig:coopFig3}
 \end{figure}
 
 \begin{table}[t!]
	\centering
	\caption{Layer ID used in Fig.~\ref{fig:coopFig3}}
	\begin{tabular}{|c|ccccc|} \hline
		playback layer &Skip& BL & EL1 & EL2 & EL3 \\ \hline
		Max Layer ID  &-1 & 0 & 1 & 2 & 3 \\ \hline
	\end{tabular}
	\label{tab : maxL}
	\vspace{-.1in}
\end{table}
 
 In Fig.~\ref{fig:coopFig2}, we observe that the $2^{nd}$ user has high bandwidth in more than half of the video length. On the other hand, the bandwidth of user 1 allows, in average, the chunks' quality to oscillate between $EL2$ and $EL3$ quality levels. That is confirmed by the results of the Off-GroupCast in Fig.~\ref{fig:coopFig3}-(d). Moreover, the figure shows that Off-GroupCast can play the first 125 seconds at $EL2$ quality and the remaining at $EL3$ with a single quality switch when the bandwidth is perfectly predicted and the buffer is infinite.
 
 Compared to the two round robin strategies, GroupCast achieves the closest performance to the Off-GroupCast. As shown in Fig.~\ref{fig:coopFig3}-(c), GroupCast delivers most of the chunks at $EL2$ and $EL3$ with more quality switches as compared to the Off-GroupCast. Fig.~\ref{fig:coopFig3}-(a, b) illustrated that both BB and PB run into skips since both of them are assigning chunks according to round robin strategy. This strategy may lead to assigning a chunk to a link that cannot fully download the chunk within its deadline.  
Fig. \ref{fig:playbuff} depicts the variation of playout buffer occupancy, and the layer decisions for each chunk. We see that the layer decisions are negatively correlated with the buffer occupancy. In conclusion, incorporating the chunk deadlines, the predicted bandwidth, and favoring the subsequent chunks make GroupCast a robust algorithm.

%


\begin{figure}[t!]
	\centering
	\includegraphics[trim=0in 0in 0in 0in, clip,  width=0.48\textwidth]{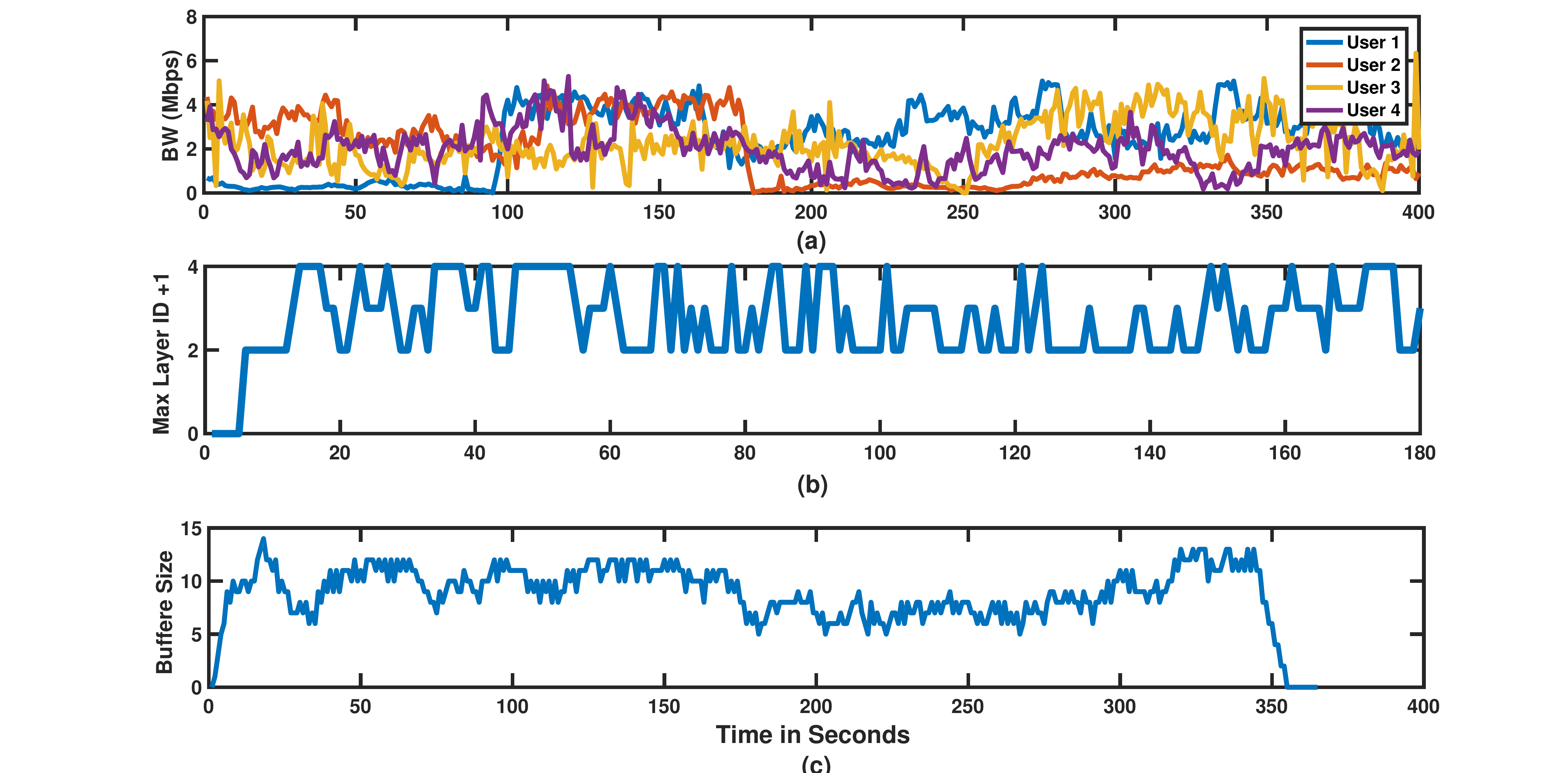} 
	\caption{Layer size and buffer occupancy when playing the video. (a) Bandwidth traces for four users, (b) the layer at which each chunk is fetched, and (c) the buffer occupancy.}
	\label{fig:playbuff}
\end{figure}

\end{thisnote}